\documentclass[journal,onecolumn,11pt,twoside]{IEEEtran}
\usepackage{epsfig,graphics,subfigure}
\usepackage{amsmath,amssymb,amsthm,multirow,balance,dsfont,hyperref}
\usepackage{cite,cleveref}

\newtheorem{theorem}{Theorem}
\newtheorem{lemma}{Lemma}

\newtheorem{assumption}{Assumption}

\usepackage{color}

\newcommand{\ba}{\boldsymbol{a}}
\newcommand{\bc}{\boldsymbol{c}}
\newcommand{\bs}{\boldsymbol{s}}

\newcommand{\bw}{\boldsymbol{w}}

\newcommand{\bx}{\boldsymbol{x}}

\newcommand{\bpsi}{\boldsymbol{\psi}}
\newcommand{\bphi}{\boldsymbol{\phi}}

\newcommand{\bH}{\boldsymbol{H}}

\newcommand{\bSig}{\boldsymbol{\Sigma}}

\newcommand{\cD}{\mathcal{D}}
\newcommand{\cE}{\mathcal{E}}

\newcommand{\cG}{\mathcal{G}}
\newcommand{\cH}{\mathcal{H}}

\newcommand{\cJ}{\mathcal{J}}
\newcommand{\cK}{\mathcal{K}}
\newcommand{\cL}{\mathcal{L}}

\newcommand{\cN}{\mathcal{N}}
\newcommand{\cP}{\mathcal{P}}
\newcommand{\cQ}{\mathcal{Q}}
\newcommand{\cT}{\mathcal{T}}

\newcommand{\cV}{\mathcal{V}}

\newcommand{\cw}{{\scriptstyle\mathcal{W}}}
\newcommand{\cx}{{\scriptstyle\mathcal{X}}}
\newcommand{\cy}{{\scriptstyle\mathcal{Y}}}

\newcommand{\bcH}{\boldsymbol{\cal{H}}}
\newcommand{\bcF}{\boldsymbol{\cal{F}}}
\newcommand{\bcw}{\boldsymbol{\cw}}

\newcommand{\wt}{\widetilde w}
\newcommand{\cwt}{\widetilde\cw}

\newcommand{\cwb}{\overline{\cw}}
\newcommand{\wb}{\overline{w}}
\newcommand{\cHb}{\overline{\mathcal{H}}}

\newcommand{\expec}{\mathbb{E}}

\newcommand{\col}{\text{col}}

\newcommand{\diag}{\text{diag}}
\newcommand{\sign}{\text{sign}}

\begin{document}

%============================
% Title.
%============================

\title{{Learning} over Multitask Graphs --\\  { Part I: Stability Analysis}}
\author{{\normalsize{Roula Nassif, \IEEEmembership{Member, IEEE}, Stefan Vlaski, \IEEEmembership{Member, IEEE}, \\
C\'edric Richard, \IEEEmembership{Senior Member, IEEE}, Ali H. Sayed, \IEEEmembership{Fellow Member, IEEE}}}\\
\thanks{The work of A. H. Sayed was supported in part by NSF grants CCF-1524250 and ECCS-1407712. {A short  version of this work appeared in the conference publication~\cite{nassif2018distributed}.}

This work was done while R. Nassif was a post-doc at EPFL. She is now with the American University of Beirut, Lebanon (e-mail: roula.nassif@aub.edu.lb). S. Vlaski and A. H. Sayed are with Institute
of Electrical Engineering, EPFL, Switzerland (e-mail: stefan.vlaski,ali.sayed@epfl.ch). C. Richard is with Universit\'e de Nice Sophia-Antipolis, France (e-mail: cedric.richard@unice.fr).
}
%\small{\linespread{0.2}Institute of Electrical Engineering, EPFL, Switzerland\\
%roula.nassif@epfl.ch\hspace{0.5cm} stefan.vlaski@epfl.ch\hspace{0.5cm} ali.sayed@epfl.ch}
}

%\title{Diffusion Adaptation over \\Multitask Graphs under Smoothness}
%\author{Roula Nassif, Stefan Vlaski, \IEEEmembership{Student Member, IEEE}, Ali H. Sayed, \IEEEmembership{Fellow Member, IEEE}\\
%\thanks{The work of A. H. Sayed was supported in part by NSF grants CCF-1524250 and ECCS-1407712. A short conference version of this work is under review~\cite{nassif2018distributed}.
%
%The authors are with Institute of Electrical Engineering, EPFL, Switzerland (e-mail: $\{$roula.nassif, stefan.vlaski, ali.sayed$\}$@epfl.ch).
%}}

\maketitle

%======================================
% Abstract.
%======================================
\begin{abstract}
This paper formulates a multitask optimization problem where agents in the network have individual objectives to meet, or individual parameter vectors to estimate, subject to a smoothness condition over the graph. The smoothness condition softens the transition in the tasks among adjacent nodes and allows incorporating information about the graph structure into the solution of the inference problem. A diffusion strategy is devised that responds to streaming data and employs stochastic approximations in place of actual gradient vectors, which are generally unavailable. The approach relies on minimizing a global cost consisting of the aggregate sum of individual costs regularized by a term that promotes smoothness. We show {in this Part I of the work}, under conditions on the step-size parameter, that the adaptive strategy induces a contraction mapping and leads to small estimation errors on the order of the small step-size. The results {in the accompanying Part II will} reveal explicitly the influence of the network topology and the regularization strength on the network performance and will provide insights into the design of effective multitask strategies for distributed inference over networks.
\end{abstract}

\begin{IEEEkeywords}
Multitask distributed inference, diffusion strategy, smoothness prior,  graph Laplacian regularization, gradient noise, {stability analysis}.
\end{IEEEkeywords}

\newpage
%======================================
% Sec: Introduction
%======================================
\section{Introduction}

Distributed inference allows a collection of interconnected agents to perform parameter estimation tasks from streaming data by relying solely on local computations and interactions with immediate neighbors. Most prior literature focuses on single-task problems, where agents with separable objective functions need to agree on a common parameter vector corresponding to the minimizer of an aggregate sum of individual costs~\cite{bertsekas1997new,olfati2007consensus,dimakis2010gossip,ram2010distributed,chen2013distributed,sayed2014adaptation,chen2015learning,chen2015learning2,sayed2014adaptive,vlaski2016diffusion}. Many network applications require more complex models and flexible algorithms than single-task implementations since their agents may need to estimate and track multiple objectives simultaneously~\cite{platachaves2017heterogeneous,chen2014multitask,nassif2016proximal,cao2017decentralized,eksin2012distributed,hallac2015network,kekatos2013distributed,platachaves2015distributed,alghunaim2017decentralized,nassif2017diffusion,chen2014diffusion}. Networks of this kind are referred to as multitask networks. Although agents may generally have distinct though related tasks to perform, they may still be able to capitalize on inductive transfer between them to improve their performance.

Based on {the type of} prior information {that may be available about how the tasks are related to each other,} multitask learning algorithms {can be} derived by translating the prior information into constraints on the parameter vectors to be inferred~\cite{platachaves2017heterogeneous,chen2014multitask,nassif2016proximal,cao2017decentralized,eksin2012distributed,hallac2015network,kekatos2013distributed,platachaves2015distributed,alghunaim2017decentralized,nassif2017diffusion,chen2014diffusion}. For example, in~\cite{kekatos2013distributed,platachaves2015distributed,alghunaim2017decentralized}, distributed strategies are developed under the assumption that the parameter vectors across the agents {overlap partially}. A more general scenario is considered in~\cite{nassif2017diffusion} where it is assumed that the tasks across the agents are locally coupled {through} linear equality constraints. In~\cite{chen2014diffusion}, the parameter space is decomposed into two orthogonal subspaces, with one of the subspaces being common to all agents. There is yet another useful way to model relationships among tasks, namely, to formulate optimization problems with appropriate regularization terms encoding {these} relationships~\cite{chen2014multitask,nassif2016proximal,cao2017decentralized,eksin2012distributed,hallac2015network}. For example, the strategy developed in~\cite{chen2014multitask} adds squared $\ell_2$-norm co-regularizers to the mean-square-error criterion to promote task similarities, while the strategy in~\cite{nassif2016proximal} adds $\ell_1$-norm co-regularizers to promote piece-wise constant transitions.

In this {paper, and the accompanying Part II~\cite{nassif2018diffusion},} we consider multitask inference problems where each agent in the network seeks to minimize an individual cost expressed as the expectation of some loss function. The minimizers of the individual costs are assumed to vary smoothly on the topology captured by the graph Laplacian matrix. The smoothness property softens the transition in the tasks among adjacent nodes and allows incorporating information about the graph structure into the solution of the inference problem. In order to exploit the smoothness prior, we formulate the inference problem as the minimization of the aggregate sum of individual costs regularized by a term  promoting smoothness, known as the graph-Laplacian regularizer~\cite{zhou2004regularization,shuman2013emerging}. A diffusion strategy is devised that responds to streaming data and employs stochastic approximations in place of actual gradient vectors, which are generally unavailable. We show {in this Part I of the work}, under conditions on the step-size learning parameter $\mu$, that the adaptive strategy induces a contraction mapping and that despite gradient noise, it is able to converge in the mean-square-error sense within $O(\mu)$ from the solution of the regularized problem, for sufficiently small $\mu$. The analysis {in the current part} also reveals how the regularization strength $\eta$ can steer the convergence point of the network toward many modes starting from the non-cooperative mode where each agent converges to the minimizer of its individual cost and ending with the single-task mode where all agents converge to a common parameter vector corresponding to the minimizer of the aggregate sum of individual costs. We shall also derive {in Part II~\cite{nassif2018diffusion}}  a closed-form expression for the steady-state network mean-square-error relative to  the minimizer of the regularized cost. This closed form expression {will reveal} explicitly the influence of the regularization strength, network topology, gradient noise, and data characteristics, on the network performance. {Additionally,} a closed-form expression for the steady-state network mean-square-error relative to  the minimizers of the individual costs {will be} {also derived} {in Part II~\cite{nassif2018diffusion}}. This expression {will provide} insights into the design of effective multitask strategies for distributed inference over networks.

There have been many works in the literature studying distributed multitask adaptive strategies and their convergence behavior. Nevertheless, with few exceptions~\cite{alghunaim2017decentralized}, most of these works {focus} on mean-square-error costs. This paper, {and the accompanying Part II~\cite{nassif2018diffusion},} generalize distributed multitask inference over networks and applies it to a wide class of individual costs. Furthermore, previous works in this domain tend to show the benefit of multitask learning empirically by simulations. Following some careful and demanding analysis, we establish in {Part II~\cite{nassif2018diffusion}, which builds on the results of this Part I,} a useful expression for the network steady-state performance. This expression provides insights into the learning behavior of multitask networks and clarifies how multitask distributed learning may improve the network performance.

\noindent\textbf{Notation.} All vectors are column vectors. Random quantities are denoted in boldface. Matrices are denoted in capital letters while vectors and scalars are denoted in lower-case letters. The operator $\preceq$ denotes an element-wise inequality; i.e., $a\preceq b$ implies that each entry of the vector $a$ is less than or equal to the corresponding entry of $b$. The symbol $\diag\{\cdot\}$ forms a matrix from block arguments by placing each block immediately below and to the right of its predecessor. The operator $\col\{\cdot\}$ stacks the column vector entries on top of each other. The symbol $\otimes$ denotes the Kronecker product. %The symbols $\otimes$ and $\otimes_b$ denote the Kronecker product and the block Kronecker product, respectively. %The symbol $\vc(\cdot)$ refers to the standard vectorization operator that stacks the columns of a matrix on top of each other and the symbol $\bvc(\cdot)$ refers to the block vectorization operation that vectorizes each block and stacks the vectors on top of each other.

%======================================
% Sec: Distributed inference under smoothness priors
%======================================
\section{Distributed inference under smoothness priors}

%============================
% Subsec: Problem formulation
%============================
\subsection{Problem formulation and adaptive strategy}
\begin{figure}
\centering
\includegraphics[scale=0.435]{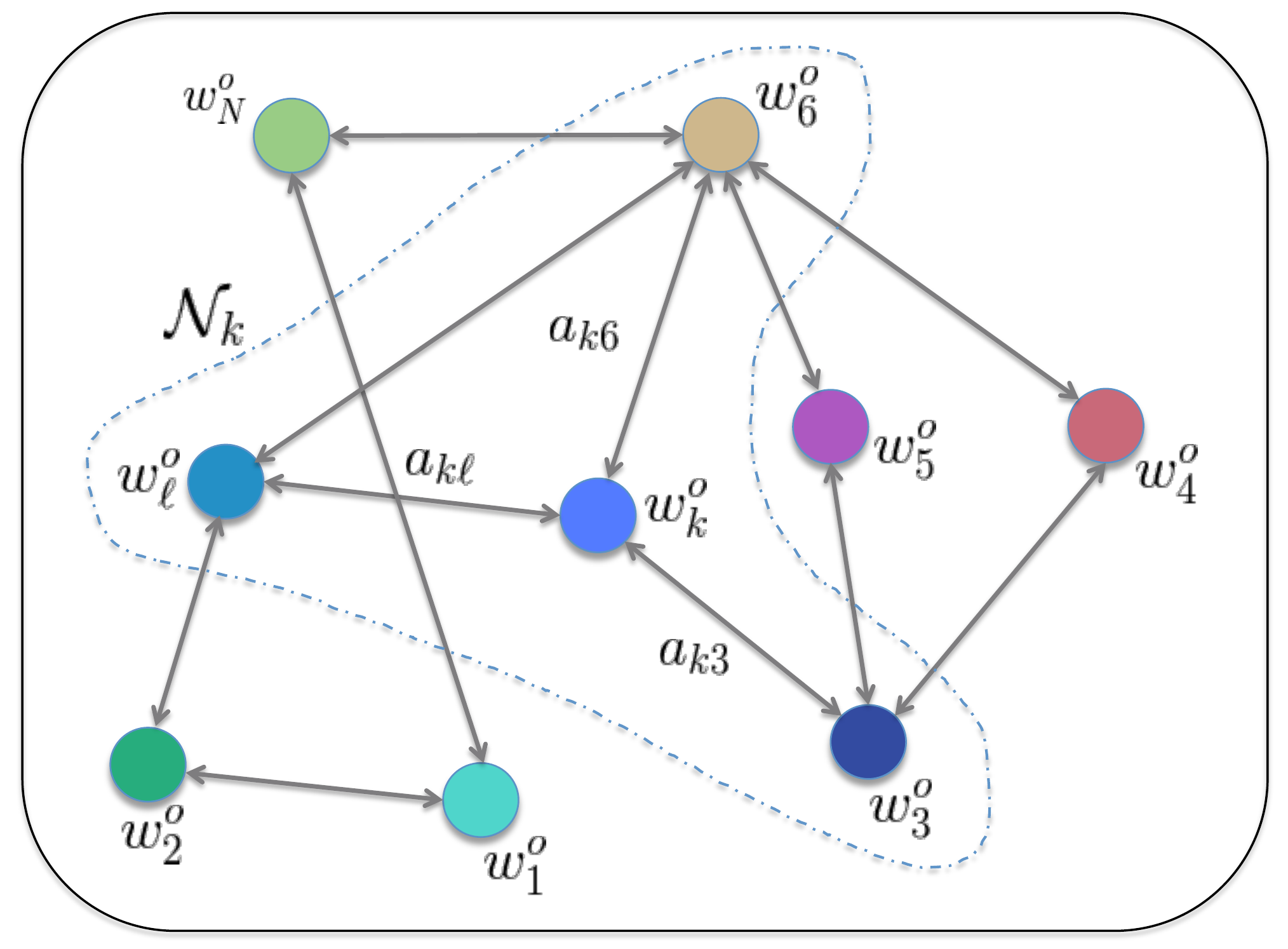}
\caption{Agents linked by an edge can share information. The weight $a_{k\ell}$ over an edge reflects the strength of the relation between $w^o_k$ at node $k$ and $w^o_\ell$ at node $\ell$.}
\label{fig: network topology}
\end{figure}
{We refer to Fig.~\ref{fig: network topology} and consider}  a connected network (or graph) $\cG=\{\cN,\cE,A\}$, where $\cN$ is a set of $N$ agents (nodes), $\cE$ is a set of edges connecting agents with particular relations, and  $A$ is a symmetric, weighted adjacency matrix. If there is an edge connecting agents $k$ and $\ell$, then $[A]_{k\ell}=a_{k\ell}>0$ reflects the strength of the relation between $k$ and $\ell$; otherwise, $[A]_{k\ell}=0$. We introduce the graph Laplacian $L$, which is a differential operator defined as $L=D-A$, where the degree matrix $D$ is a diagonal matrix with $k$-th entry $[D]_{kk}=\sum_{\ell=1}^Na_{k\ell}$. Since $L$ is symmetric positive semi-definite, it possesses a complete set of orthonormal eigenvectors. We denote them by $\{v_1,\ldots,v_N\}$. For convenience, we order the set of real, non-negative eigenvalues of $L$ as $0=\lambda_1<\lambda_2\leq\ldots\leq\lambda_N=\lambda_{\max}(L)$, where, since the network is connected, there is only one zero eigenvalue with corresponding eigenvector $v_1=\frac{1}{\sqrt{N}}\mathds{1}_N$~\cite{chung1997spectral}. Thus, the Laplacian can be decomposed as:
\begin{equation}
\label{eigendecomposition of the Laplacian}
L=V\Lambda V^\top,
\end{equation}
where $\Lambda=\diag\{\lambda_1,\ldots,\lambda_N\}$ and $V=[v_1,\ldots,v_N]$.

Let $w_k\in\mathbb{R}^M$ denote some parameter vector at agent $k$ and let $\cw=\col\{w_1,\ldots,w_N\}$ denote the collection of parameter vectors from across the network. We associate with each agent $k$ a risk function $J_k(w_k):\mathbb{R}^M\rightarrow\mathbb{R}$ assumed to be strongly convex. In most learning and adaptation problems, the risk function is expressed as the expectation of a loss function $Q_k(\cdot)$ and is written as $J_k(w_k)=\expec\,Q_k(w_k;\bx_k)$, where $\bx_k$ denotes the random data. The expectation is computed over the distribution of this data. We denote the unique minimizer of $J_k(w_k)$ by $w^o_k$.  {We introduce {a common} assumption on the risks $\{J_k(w_k)\}$. This condition is applicable to many situations of {interest (see, e.g.,~\cite{sayed2014adaptation,sayed2014adaptive})}.
\begin{assumption}{\rm(Strong convexity)}
\label{assumption: strong convexity}
 It is assumed that the individual costs $J_k(w_k)$ are each twice differentiable and strongly convex such that the Hessian matrix function $H_k(w_k)=\nabla^2_{w_k}J_k(w_k)$ is uniformly bounded from below and {above, say, as:}
\begin{equation}
0<\lambda_{k,\min}I_M\leq H_k(w_k)\leq\lambda_{k,\max}I_M,
\end{equation}
where $\lambda_{k,\min}>0$ for $k=1,\ldots,N$.
\qed
\end{assumption}
}

In many situations, there is prior information available about $\cw^o=\col\{w_1^o,\ldots,w_N^o\}$. In the current work, the prior belief we want to enforce is that the target signal $\cw^o$ is smooth with respect to the underlying weighted graph. References~\cite{chen2014multitask,nassif2016proximal,cao2017decentralized} provide variations for such problems for the special case of mean-square-error costs. {Here we treat general convex costs. Let} $\cL= L\otimes I_M$. The smoothness of $\cw$ can be measured in terms of a quadratic form of the graph Laplacian~\cite{zhou2004regularization,shuman2013emerging,weinberger2007graph}:
\begin{equation}
\label{eq: quadratic regularization}
S(\cw)=\cw^\top\cL \cw=\frac{1}{2}\sum_{k=1}^N\sum_{\ell\in\cN_k}a_{k\ell}\|w_k-w_{\ell}\|^2,
\end{equation}
where $\cN_k$ is the set of neighbors of $k$, i.e., the set of nodes connected to agent $k$ by an edge. {Figure~\ref{fig: network topology} provides an illustration.} The smaller $S(\cw)$ is, the  smoother the  signal $\cw$ on the  graph is. Intuitively, given that the weights are non-negative, $S(\cw)$ shows that $\cw$ is  considered  to  be  smooth if nodes with a large $a_{k\ell}$ on the edge connecting them have similar weight values $\{w_k,w_{\ell}\}$. Our objective is to devise and study a strategy that solves the following regularized problem:
\begin{equation}
\label{eq: global problem}
\cw^o_{\eta}=\arg\min_{\cw}J^{\text{glob}}(\cw)=\sum_{k=1}^NJ_k(w_k)+\frac{\eta}{2}\, \cw^{\top}\cL \cw,
\end{equation}
in a distributed manner where each agent is interested in estimating the $k$-th sub-vector of $\cw^o_{\eta}=\col\{w_{1,\eta}^o,\ldots,w_{N,\eta}^o\}$. The tuning parameter $\eta\geq 0$ controls the trade-off between the two components  of the objective function. Reference~\cite{nassif2018distributed} provides a theoretical motivation for the optimization framework where it is shown {that, under a Gaussian Markov random field assumption,} solving problem~\eqref{eq: global problem} {is} equivalent to finding a maximum a posteriori (MAP) estimate for $\bcw$.  We are particularly interested in solving the problem in the stochastic setting when the distribution of the data $\bx_k$ in  $J_k(w_k)=\expec\,Q_k(w_k;\bx_k)$ is generally unknown. This means that the risks $J_k(w_k)$ and their gradients $\nabla_{w_k}J_k(w_k)$ are unknown. As such, approximate gradient vectors need to be employed. A common construction in stochastic approximation theory is to employ the following approximation at  iteration $i$:
\begin{equation}
\widehat{\nabla_{w_k}J_k}(w_k)=\nabla_{w_k}Q_k(w_k;\bx_{k,i}),
\end{equation}
where $\bx_{k,i}$ represents the data observed at iteration $i$. The difference between the true gradient and its approximation is called the gradient noise $\bs_{k,i}(\cdot)$:
\begin{equation}
\label{eq: gradient noise process}
\bs_{k,i}(w)\triangleq\nabla_{w_k}J_k(w)-\widehat{\nabla_{w_k}J_k}(w).
\end{equation}
Each agent can employ a stochastic gradient descent update to estimate $w^o_{k,\eta}$:
\begin{equation}
\label{eq: update 1}
\bw_{k,i}=\bw_{k,i-1}-\mu\widehat{\nabla_{w_k}J_k}(\bw_{k,i-1})-\mu\eta\sum_{\ell\in\cN_k}a_{k\ell}(\bw_{k,i-1}-\bw_{\ell,i-1}),
\end{equation}
where $\mu>0$ is a small step-size parameter. In this implementation, each agent $k$ collects from its neighbors the estimates $\bw_{\ell,i-1}$, and performs a stochastic-gradient descent update on:
\begin{equation}
\bar{J}_{k,i-1}(w_k)\triangleq J_k(w_k)+\frac{\eta}{2}\sum_{\ell\in\cN_k}a_{k\ell}\|w_k-\bw_{\ell,i-1}\|^2.
\end{equation}
By introducing an auxiliary variable $\bpsi_{k,i}$, strategy~\eqref{eq: update 1} can be implemented in an incremental manner:
\begin{equation}
\label{eq: distributed algorithm}
\left\lbrace
\begin{array}{lr}
\bpsi_{k,i}=\bw_{k,i-1}-\mu\widehat{\nabla_{w_k}J_k}(\bw_{k,i-1})\\
\bw_{k,i}=\bpsi_{k,i}-\mu\eta\displaystyle\sum_{\ell\in\cN_k}a_{k\ell}(\bpsi_{k,i}-\bpsi_{\ell,i}),
\end{array}
\right.
\end{equation}
where we replaced $(\bw_{k,i-1}-\bw_{\ell,i-1})$ in the second step by the difference $(\bpsi_{k,i}-\bpsi_{\ell,i})$ since we expect $\bpsi_{k,i}$ to be an improved estimate compared to $\bw_{k,i-1}$. {Note that if we introduce the coefficients:
\begin{equation}
c_{k\ell}=\left\lbrace
\begin{array}{ll}
1-\mu\eta\displaystyle\sum_{\ell\in\cN_k}a_{k\ell},&k=\ell\\
\mu\eta a_{k\ell},&\ell\in\cN_k\setminus \{k\}\\
0,&\ell\notin\cN_k
\end{array}\right.
\end{equation}
then recursion~\eqref{eq: distributed algorithm} can be written in the diffusion form\cite{chen2013distributed,sayed2014adaptation,chen2015learning,chen2015learning2,sayed2014adaptive}:
\begin{equation}
\label{eq: distributed diffusion algorithm}
\left\lbrace
\begin{array}{lr}
\bpsi_{k,i}=\bw_{k,i-1}-\mu\widehat{\nabla_{w_k}J_k}(\bw_{k,i-1})\\
\bw_{k,i}=\displaystyle\sum_{\ell\in\cN_k}c_{k\ell}\bpsi_{\ell,i},
\end{array}
\right.
\end{equation}
where the second step is a combination step. If we collect the scalars $\{c_{k\ell}\}$ into the matrix $C=[c_{k\ell}]$, then the entries of $C$ are non-negative for small enough $\mu$ and its columns and rows add up to one, i.e., $C$ is a doubly-stochastic matrix. We shall continue with form~\eqref{eq: distributed algorithm} because the second step in~\eqref{eq: distributed algorithm} makes the dependence on $\eta$ explicit. We will show later that by varying the value of $\eta$ we can make the algorithm behave in different ways from {fully non-cooperative} to fully single-task with many other modes in between.}

%===============================
% Subsec: Summary of main results
%===============================
\subsection{{Summary of main results}}
{Before delving into the study of the learning capabilities of~\eqref{eq: distributed algorithm} and its performance limits, we summarize in this section, for the benefit of the reader, the main conclusions of {this Part I, and its accompanying Part II~\cite{nassif2018diffusion}}. One key insight that will follow from the detailed analysis {in this Part I} is that the smoothing parameter $\eta$ can be regarded as an effective tuning parameter that controls the nature of the learning process. The value of $\eta$ can vary from $\eta=0$ to $\eta\rightarrow\infty$. We will show that at one end, when $\eta=0$, the learning algorithm reduces to a non-cooperative mode of operation where each agent acts individually and estimates its own local model, $w^o_k$. On the other hand, when $\eta\rightarrow\infty$, the learning algorithm moves to a single-mode of operation where all agents cooperate to estimate a {\em single} parameter (namely, the Pareto solution of the aggregate cost function). For any values of $\eta$ in the range $0<\eta<\infty$, the network behaves in a multitask mode where agents seek their individual models while at the same time ensuring that these models satisfy certain smoothness and closeness conditions dictated by the value of $\eta$. We are not only interested in a qualitative description of the network behavior. Instead, we would like to characterize these models in a quantitative manner by deriving expressions that allow us to predict performance as a function of $\eta$ and, therefore, fine tune the network to operate in different scenarios.}

{To begin with, recall that the objective of the multitask strategy~\eqref{eq: distributed algorithm} is to exploit similarities among neighboring agents in an attempt to improve the overall network performance in approaching the collection of individual minimizer ${\cw}^o$ by means of local communications. In light of the fact that algorithm~\eqref{eq: distributed algorithm} has been derived as an (incremental) gradient descent recursion for the regularized cost~\eqref{eq: global problem}, whose minimizer \( \cw_{\eta}^o \) is in general different from  \( \cw^o \), the limiting point of algorithm~\eqref{eq: distributed algorithm} will therefore be generally different from \( \cw^o \), the actual objective of the multitask learning problem. This mismatch is the ``cost'' of enforcing smoothness. The analysis in the paper will reveal that the mismatch is a function of the similarity between the individual minimizers $\{w_k^o\}$, of second-order properties of the individual costs, of the network topology captured by \( L \), and of the regularization strength \( \eta \).  In particular, future expression~\eqref{eq: GFT of w o eta - wo} will allow us to understand the interplay between these quantities which is important  for the design of effective multitask strategies. The key conclusion will be that, while the bias (difference between $\cw_{\eta}^o$ and $\cw^o$) will  in general increase as the regularization strength \( \eta \) increases, the \emph{size} of this increase is determined by the smoothness of $\cw^o$ which is in turn  function of the network topology  captured by \( L \). The more similar the tasks at neighboring agents are, the smaller the bias will be. This result, while intuitive, is reassuring, as it implies that as long as  $\cw^o$ is sufficiently smooth, the bias induced by regularization will remain small, even for moderate regularization strengths \( \eta \).  
}

{The analysis also quantifies the \emph{benefit} of cooperation, namely, the objective of improving the mean-square deviation around the limiting point of the algorithm. This analysis is challenging due to coupling among agents, and the multi-task nature of the learning process (where agents have individual targets but need to meet certain smoothness and closeness conditions with their neighbors). Section~\ref{sec: network stability} in this Part I and Sections III and IV in Part II~\cite{nassif2018diffusion}, and the supporting appendices, are devoted to carrying out this analysis in depth leading, for example to Theorem~1 in Part II~\cite{nassif2018diffusion}. This theorem gives expressions for the mean-square-deviation (MSD) relative to \( \cw_{\eta}^o \). The expressions reveal the effect of the step-size parameter $\mu$, regularization strength $\eta$, network topology, and data characteristics (captured by the smoothness profile, second-order properties of the costs, and second-order moments of the gradient noise) on the size of the  steady-state  mean-square-error performance. The results established in Theorem~1 and expression~(82) in Part II~\cite{nassif2018diffusion} provide  tools for characterizing the performance of multitask strategies in some great detail.}

%{Section~\ref{sec: network stability} in this Part I and Sections III and IV in Part II~\cite{nassif2018diffusion} are devoted to carrying out this analysis in depth leading, for example to Theorem~1 in Part II~\cite{nassif2018diffusion}. In Section~\ref{sec: network stability} in this Part I, and the supporting appendices, we show, under some conditions on the step-size parameter, that algorithm~\eqref{eq: distributed algorithm} induces a contraction mapping and leads to small mean-square estimation errors on the order of the small step-size. We shall also establish in this part the stability of the fourth and first-order moments of the error vector; these steps are required to establish the steady-state mean-square-error performance in the following Part II~\cite{nassif2018diffusion}. {Sections III and IV} in Part II~\cite{nassif2018diffusion}, and the supporting appendices, characterize how close algorithm~\eqref{eq: distributed algorithm} converges to the shifted limiting point \( \cw_{\eta}^o \).} Theorem~1 in Part II~\cite{nassif2018diffusion} 

{To illustrate the power of these results, consider a connected network where each agent is subjected to streaming data. The goal at each agent is to estimate a local parameter vector $w^o_k$ from the observed data by minimizing a cost of the form $J_k(w_k)=\expec\,Q_k(w_k;\bx_k)$, where $\bx_k$ denotes the random data. Consider network  applications where the  minimizers at neighboring agents tend to be similar~\cite{chen2014multitask,shuman2013emerging}. Although each agent is interested in estimating its own task $w^o_k$, cooperating neighboring agents can still benefit from their interactions because of this closeness.  Given the graph Laplacian and data characteristics, one problem of interest would be to determine the optimal cooperation rule, i.e., the value of $\eta$ that minimizes the network mean-square-error performance. Future {expression~(82) in Part II~\cite{nassif2018diffusion}} can be used to solve this problem since it allows us to predict the network MSD relative to $\cw^o$. By using {expression~(82) in Part II~\cite{nassif2018diffusion}}, for example, we will be able to construct curves of the form shown in Fig.~\ref{fig: network MSD bar}, which illustrate how performance is dependent on the smoothness parameter $\eta$ and how the nature of the limiting solution varies as a function of this parameter.  As it can be seen from this figure, $\eta=4$ gives the best network steady-state mean-square performance. Note that $\eta=0$ corresponds to the non-cooperative scenario and that a large $\eta$ induces a large bias in the estimation. In the sequel we will show that as $\eta$ varies from $\eta=0$ to $\eta\rightarrow\infty$, the network behavior moves from the non-cooperative mode of operation (where agents act independently) to the single-task mode of operation (where all agents focus on estimating a single parameter). For values of $\eta$ in between, the network can operate in any of a multitude of multitask modes (where agents estimate their own local parameters under smoothness conditions to allow for some similarity between adjacent nodes). These limits are indicated in Fig.~\ref{fig: network MSD bar}. }
\begin{figure}
\centering
\includegraphics[scale=0.5]{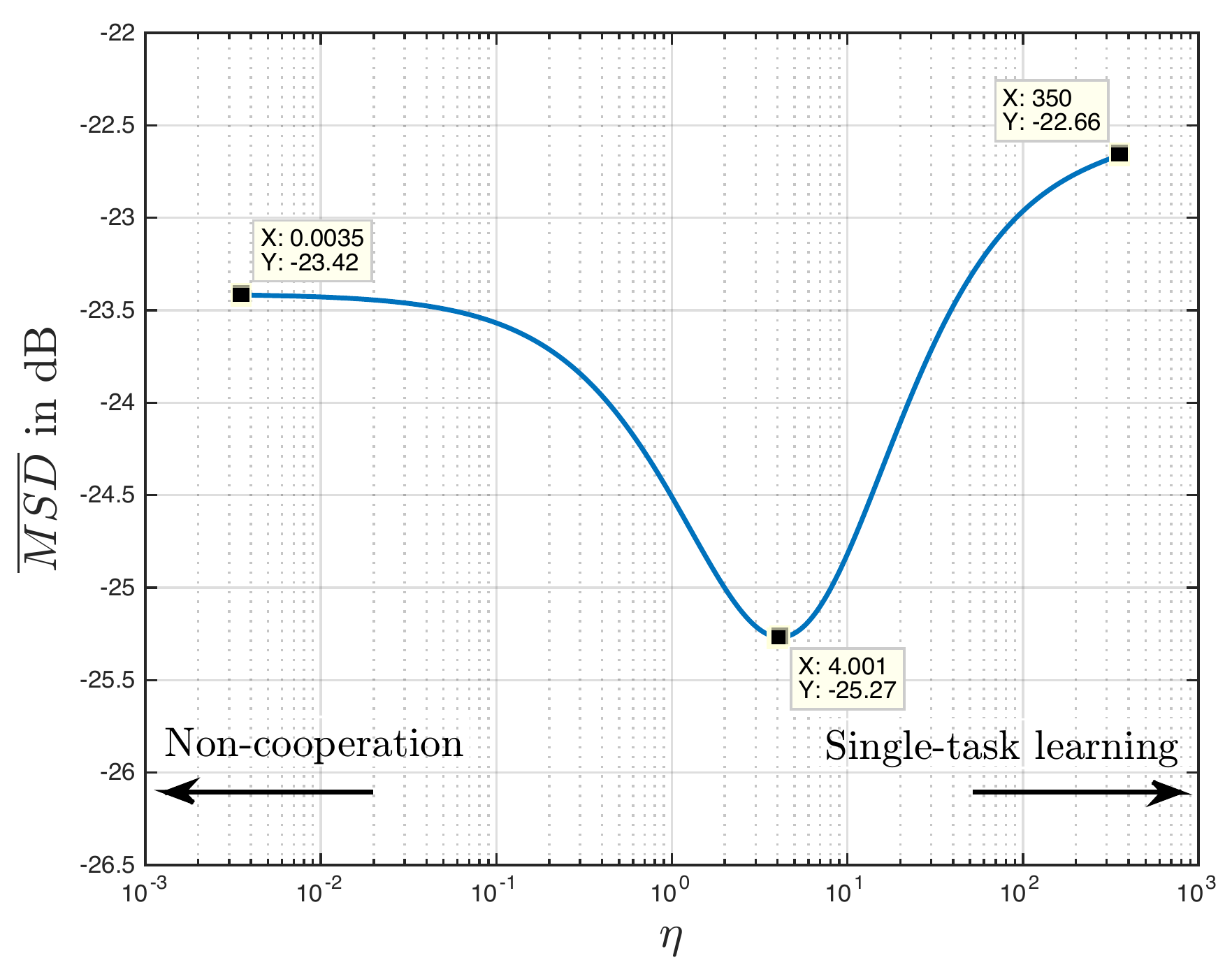}
\caption{Network steady-state MSD relative to a smooth signal $\cw^o$ as a function of the regularization strength $\eta\in[0,350]$ at $\mu=0.005$.}
\label{fig: network MSD bar}
\end{figure}
%{The remainder of the manuscript is dedicated to the study of the possible \emph{benefit} of cooperation, namely the objective of improving the mean-square deviation around the limiting point of the algorithm. It ultimately leads to expressions for the mean-square-deviation relative to the shifted limiting point \( w_{\eta}^o \) in Theorem~\ref{lemma: steady-state MSD performance}. The expressions reveal the effect of regularization and communication on the perturbations resulting from the gradient noise in the stochastic gradient implementation and clarify the positive and negative effects of encouraging smoothness in the solution as a function of the regularization parameters, smoothness profile and data profile across the network. Combining these expressions with the study on the difference between \( \cw^o \) and \( \cw_{\eta}^o \) in Lemma~\ref{lemm: transformed w o eta}, closes the loop and allows for the prediction of performance of each individual agent relative to the multi-task objective \( \cw_{\eta}^o \).}

{Finally, we would like to mention that one of the main tools used in the analysis in this work, and its accompanying Part II~\cite{nassif2018diffusion}, is the linear transformation relative to the eigenspace of the graph Laplacian $L$ from~\cite{chen2013distributed}, which is also known as the graph Fourier transform~\cite{shuman2013emerging,ortega2018graph,tsitsvero2016signals}. Under some conditions on the data and costs profile, we show in {Section~VI-A in Part II~\cite{nassif2018diffusion}} how the diffusion type algorithm~\eqref{eq: distributed algorithm} exhibits a low-pass graph filter behavior. Such  filters are commonly used to reduce the network noise profile when the signal to be estimated is smooth with respect to the underlying topology~\cite{shuman2013emerging,sandryhaila2013discrete,chen2014signal,shuman2011chebyshev}. Interestingly, the theoretical results established in {this Part I, and its accompanying Part II~\cite{nassif2018diffusion},} reveal the reasons for performance improvements under localized cooperation.}

%==================================
% Subsec: Network limit point and regularization strength
%==================================
\subsection{Network limit point and regularization strength}
Before examining the behavior and performance of strategy~\eqref{eq: distributed algorithm} with respect to the limiting point $\cw^o_{\eta}$ in~\eqref{eq: global problem}, we discuss the influence of $\eta$ on $\cw^o_{\eta}$. When $\eta=0$, we have {from~\eqref{eq: global problem} that} $\cw^o_{\eta}=\cw^o$ and strategy~\eqref{eq: distributed algorithm} reduces to the single-agent mode of operation or the non-cooperative solution where each agent minimizes $J_k(w_k)$ locally without cooperation. When $\eta\rightarrow\infty$, we have {from~\eqref{eq: global problem} that} $\cw^o_{\eta}=\mathds{1}_N\otimes w^\star$ {where}
\begin{equation}
\label{eq: w star}
w^\star\triangleq\arg\min_{w}\sum_{k=1}^NJ_k(w),
\end{equation}
and we are in the single-task mode of operation where all agents seek to estimate a common parameter vector $w^\star$ corresponding to the minimizer of the aggregate sum of individual costs~\cite{sayed2014adaptation,chen2013distributed,chen2015learning,chen2015learning2,sayed2014adaptive}. In order to study more closely the influence of (finite) $\eta>0$ on the network output $\cw^o_{\eta}$, we  {examine} the influence of $\eta$ on the transformed vector:
\begin{equation}
\label{eq: w bar o eta definition}
\cwb^o_\eta\triangleq(V^\top\otimes I_M)\cw^o_{\eta}=\col\left\{{\wb^o_{m,\eta}}\right\}_{m=1}^N,
\end{equation}
{with the $m$-th sub-vector $\wb^o_{m,\eta}$}  {denoting the} spectral content of $\cw^o_\eta$ at the $m$-th eigenvalue $\lambda_m$ of the Laplacian:
\begin{equation}
{\wb^o_{m,\eta}}=(v_m^\top\otimes I_M)\cwb^o_\eta.
\end{equation}
%The Laplacian matrix $\cL$ can be decomposed as $\cL= L\otimes I_M=\cV\cJ\cV^\top$ where $\cV=V\otimes I_M$ and $\cJ=\Lambda\otimes I_M$. Let $\cwb=\cV^\top\cw=\col\{\left[\cwb\right]_1,\ldots,\left[\cwb\right]_N\}$ where $\left[\cwb\right]_m=(v_m^\top\otimes I_M)\cw$.
From~\eqref{eigendecomposition of the Laplacian}, the quadratic {regularization term} $S(\cw)$ in~\eqref{eq: quadratic regularization} can be written as:
\begin{equation}
\label{eq: S(w) smooth}
S(\cw)=\cw^\top\cL \cw=\sum_{m=1}^N\lambda_m\|{\wb_{m}}\|^2=\sum_{m=2}^N\lambda_m\|{\wb_{m}}\|^2,
\end{equation}
where ${\wb_{m}}=(v_m^\top\otimes I_M)\cw$ and where we used the fact that $\lambda_1=0$. Intuitively, given that $\lambda_m>0$ for $m=2,\ldots, N$, the above expression shows that $\cw$ is considered to be smooth if $\|{\wb_{m}}\|^2$ corresponding to large $\lambda_m$ is small. {As a result,} for a fixed $\lambda_m>0$, {and as} the regularization strength $\eta>0$ in~\eqref{eq: global problem} {increases, one would expect $\|{\wb^o_{m,\eta}}\|^2$} {to decrease. Similarly,} for a fixed $\eta\geq 0$, {and as} $\lambda_m>0$ {increases, one would expect} $\|{\wb^o_{m,\eta}}\|^2$ {to decrease as well. However, as} we will see in the {sequel}, this {behavior does not always hold. We show in {Section~VI-A in Part II~\cite{nassif2018diffusion}} that this is valid} when the Hessian matrix function $H_k(w_k)\triangleq\nabla^2_{w_k}J_k(w_k)$ is independent of $w_k$, {i.e., the cost $J_k(w_k)$ is quadratic in $w_k$} and is uniform across the network. For more general scenarios, this is not necessarily the {case. What is useful to note, however, is that as $\eta$ moves from $0$ towards $\infty$, a variety of solution points $\cw^o_{\eta}$ can occur ranging from the non-cooperative to the single-task solution at both extremes.}

\noindent From the optimality condition of~\eqref{eq: global problem}, we have:
\begin{equation}
\label{eq: optimality condition}
\col\left\{\nabla_{w_k}J_k(w_{k,\eta}^o)\right\}_{k=1}^N=-\eta\cL \cw^o_{\eta},
\end{equation}
Using the mean value theorem~\cite[pp.~24]{polyak1987introduction}, we can write:
\begin{equation}
\label{eq: mean value theorem}
\nabla_{w_k}J_k(w_{k,\eta}^o)-{\underbrace{\nabla_{w_k}J_k(w_{k}^o)}_{=0}}=H^o_{k,\eta}(w_{k,\eta}^o-w_{k}^o),
\end{equation}
where
\begin{equation}
H^o_{k,\eta}\triangleq\int_{0}^1\nabla^2_{w_k}J_k(w_{k}^o+t(w_{k,\eta}^o-w_{k}^o))dt.
\end{equation}
Let $\cH^o_{\eta}\triangleq\diag\left\{H^o_{k,\eta}\right\}_{k=1}^N$. Relation~\eqref{eq: optimality condition} can {then be rewritten more compactly as:}
\begin{equation}
\label{eq: w o eta}
\cw^o_{\eta}=\left( \cH^o_{\eta}+\eta\cL \right)^{-1} \cH^o_{\eta}\cw^o.
\end{equation}
Note that the inverse in~\eqref{eq: w o eta} {exists for all $\eta\geq 0$} since the matrix $\cL$ is positive semi-definite {and,}  under Assumption~\ref{assumption: strong convexity}, the matrix $\cH_{\eta}^o$ is positive definite. Pre-multiplying both sides of the above relation by $(V\otimes I_M)^\top$ gives:
\begin{equation}
\label{eq: GFT of w o eta}
\cwb^o_{\eta}=(\cHb_{\eta}^o+\eta\cJ)^{-1}\cHb^o_{\eta}\cwb^o,
\end{equation}
where {$\cwb^o_{\eta}$ is defined in~\eqref{eq: w bar o eta definition}}, $\cwb^o\triangleq\cV^\top \cw^o$, $\cV\triangleq V\otimes I_M$, and
\begin{align}
%\cwb^o&\triangleq\cV^\top \cw^o,\\
%\cV&\triangleq V\otimes I_M,\\
\cJ&\triangleq\Lambda\otimes I_M,\label{eq: definition cJ}\\
\cHb^o_{\eta}&\triangleq\cV^\top\cH^o_{\eta}\cV.\label{eq: definition cH o eta}
\end{align}
 Since $L$ has a single eigenvalue at zero, $\Lambda$ and $V$ can be partitioned as follows:
\begin{equation}
\Lambda=\diag\{0,\Lambda_o\},~~V=[v_1,V_R],  ~~\text{and} ~~ V^\top=\col\{v^\top_1,V_R^\top\},.
\end{equation}

\begin{lemma}{\rm (Limiting point)}
\label{lemm: transformed w o eta}
Under Assumption~\ref{assumption: strong convexity}, it can be shown that $\cwb^o_{\eta}$ given by~\eqref{eq: GFT of w o eta} {satisfies}:
\begin{equation}
\label{eq: GFT of w o eta 2}
\cwb^o_{\eta}=\left[
\begin{array}{lc}
I_M&\cQ_{11}^{-1}\cQ_{12}\left(I_{M(N-1)}-\cK\right)\\
0&  \cK
\end{array}
\right]
\left[
\begin{array}{lr}
{\wb^o_{1}}\\
\left[\cwb^o\right]_{2:N}
\end{array}
\right],
\end{equation}
where ${\wb^o_{1}}=(v_1^\top\otimes I_M)\cw^o$, $\left[\cwb^o\right]_{2:N}=(V_R^\top\otimes I_M)\cw^o$ and
\begin{align}
\cQ_{11}&\triangleq(v_1^\top\otimes I_M)\cH^o_{\eta}(v_1\otimes I_M)=\frac{1}{N}\sum_{k=1}^NH^o_{k,\eta},\label{eq: Q 11 lemma}\\
\cQ_{12}&\triangleq(v_1^\top\otimes I_M)\cH^o_{\eta}(V_R\otimes I_M),\label{eq: Q 12 lemma}\\
\cQ_{22}&\triangleq(V_R^\top\otimes I_M)\cH^o_{\eta}(V_R\otimes I_M)+\eta\Lambda_o\otimes I_M,\label{eq: Q 22 lemma}\\
\cG&\triangleq(\cQ_{22}-\cQ_{12}^\top\cQ_{11}^{-1}\cQ_{12})^{-1},\label{eq: G 22 lemma}\\
\cK&\triangleq I_{M(N-1)}-\eta\,\cG\left(\Lambda_o\otimes I_M\right)\label{eq: K 22 lemma 1}
\end{align}
\end{lemma}
\begin{proof}
See Appendix~\ref{sec: Proof of Lemma transformed w o eta} where {we also show that}:
\begin{equation}
\label{eq: order of K 22}
\|\cK\|\leq\left(\max_{1\leq k \leq N}\lambda_{k,\max}\right)\left(\eta\lambda_2(L)+\min_{1\leq k \leq N}\lambda_{k,\min}\right)^{-1}=\frac{O(1)}{(O(1)+O(\eta))}.
\end{equation}
\end{proof}
\noindent {Consider the difference between $\cw^o_{\eta}$ and $\cw^o$. It turns out that the smoother $\cw^o$ is, the smaller $\|\cw^o-\cw^o_{\eta}\|$ will be. To see this, let us subtract {$\cwb^o$ from both sides of equation~\eqref{eq: GFT of w o eta 2}}. We obtain:
\begin{equation}
\label{eq: GFT of w o eta - wo}
\cwb^o_{\eta}-\cwb^o=\left[
\begin{array}{c}
 \cQ_{11}^{-1}\cQ_{12}\left( I_{M(N-1)}-\cK\right)\\
\cK-I_{M(N-1)}
\end{array}
\right]
\left[\cwb^o\right]_{2:N}.
\end{equation}
The difference $\cwb^o_{\eta}-\cwb^o$ depends on $\left[\cwb^o\right]_{2:N}$. {Thus, from~\eqref{eq: S(w) smooth} and~\eqref{eq: GFT of w o eta - wo}, we conclude that} the smoother $\cw^o$ is, the smaller $\|\cw^o_{\eta}-\cw^o\|=\|\cwb^o_{\eta}-\cwb^o\|$ will be.  }

{Lemma~\ref{lemm: transformed w o eta} will be useful in the sequel to establish Theorem~\ref{theo: dimension of the bias}  and to provide a low-pass graph filter interpretation for the uniform Hessian matrices scenario considered in {Section~VI-A in Part II~\cite{nassif2018diffusion}}.}
\section{Network stability}
\label{sec: network stability}
We examine the behavior of algorithm~\eqref{eq: distributed algorithm} under {Assumption~\ref{assumption: gradient noise}}  on the gradient noise processes $\{\bs_{k,i}(\cdot)\}$ defined in~\eqref{eq: gradient noise process}. As explained in~\cite{sayed2014adaptation,sayed2014adaptive}, these conditions are {automatically satisfied in many situations of interest} in learning and adaptation. {Condition~\eqref{eq: mean gradient noise condition} essentially states that the gradient vector approximation should be unbiased conditioned on the past data, which is a reasonable condition to require. Condition~\eqref{eq: condition on second-order moment of gradient noise} states that the second-order moment of the gradient noise process should get smaller for better estimates, since it is bounded by the squared norm of the iterate. Condition~\eqref{eq: uncorrelated gradient noises} states that the gradient noises across the agents are uncorrelated.}
\begin{assumption}
\label{assumption: gradient noise}
\emph{(Gradient noise process)} The gradient noise process defined in~\eqref{eq: gradient noise process} satisfies for any $\bw\in\bcF_{i-1}$ {and for all $k,\ell=1,2,\ldots,N$:
\begin{align}
\expec[\bs_{k,i}(\bw)|\bcF_{i-1}]&=0,\label{eq: mean gradient noise condition}\\
\expec[\|\bs_{k,i}(\bw)\|^2|\bcF_{i-1}]&\leq\beta^2_k\|\bw\|^2+\sigma^2_{s,k},\label{eq: condition on second-order moment of gradient noise}\\
\expec[\bs_{k,i}(\bw)\bs_{\ell,i}^\top(\bw)|\bcF_{i-1}]&=0,\quad k\neq \ell,\label{eq: uncorrelated gradient noises}
\end{align}
for some} $\beta^2_k\geq 0$,  $\sigma^2_{s,k}\geq 0$, and where $\bcF_{i-1}$ denotes the filtration generated by the random processes $\{\bw_{\ell,j}\}$ for all $\ell=1,\ldots,N$ and $j\leq i-1$.
\qed
\end{assumption}

In this section, we analyze how well the multitask strategy~\eqref{eq: distributed algorithm} approaches the optimal solution $\cw^o_{\eta}$ of the regularized cost~\eqref{eq: global problem}. We examine this performance in terms of the mean-square-error measure, $\expec\|w^o_{k,\eta}-\bw_{k,i}\|^2$, the fourth-order moment, $\expec\|w^o_{k,\eta}-\bw_{k,i}\|^4$, and the mean-error process, $\expec (w^o_{k,\eta}-\bw_{k,i})$. To establish mean-square error stability, we extend the energy analysis framework of~\cite{chen2013distributed} to handle multitask distributed optimization. Then, following a similar line of reasoning as in~\cite[Chapter 9]{sayed2014adaptation}, we establish the {stability of the first and fourth-order moments, which is} necessary to {arrive at an expression for the} steady-state performance {in Part II~\cite{nassif2018diffusion}}. %{Due to space limitations, apart from the main theorems and lemmas statements, we include only the proofs related to the mean-square-error stability theorem.  For the remaining statements, the proofs can be found in the technical report~[x].}

Let us introduce the network block vector $\bcw_i=\col\{\bw_{1,i},\ldots,\bw_{N,i}\}$. At each iteration, we can view~\eqref{eq: distributed algorithm} as a mapping from $\bcw_{i-1}$ to $\bcw_i$:
\begin{equation}
\label{eq: network vector recursion}
\boxed{\bcw_i=\left(I_{MN}-\mu\eta\cL\right)\left(\bcw_{i-1}-\mu\,\col\left\{\widehat{\nabla_{w_k}J_k}(\bw_{k,i-1})\right\}_{k=1}^N\right)}
\end{equation}
{We introduce the following condition} on the combination matrix $(I_{MN}-\mu\eta\cL)$, which is necessary for studying the performance {of~\eqref{eq: distributed algorithm}. It can be easily verified that this requirement is always met by selecting $\mu$ and $\eta$ to satisfy the bounds~\eqref{eq: condition for stability}--\eqref{eq: condition for positivity}.}
\begin{assumption}
\label{assumption: combination matrix}
\emph{(Combination matrix)} The symmetric combination matrix $\left(I_{MN}-\mu\eta\cL\right)$ {has nonnegative entries and its spectral radius is equal to one. Since $L$ has an eigenvalue at zero, these} conditions are satisfied when the step-size $\mu>0$ and the regularization strength $\eta\geq 0$ satisfy:
\begin{align}
\label{eq: condition for stability}&0\leq\mu\eta\leq\frac{2}{\lambda_{\max}(L)},\\
\label{eq: condition for positivity}&0\leq\mu\eta\leq\min_{1\leq k\leq N}\left\{\frac{1}{\sum_{\ell=1}^Na_{k\ell}}\right\},
\end{align}
where condition~\eqref{eq: condition for stability} ensures stability and condition~\eqref{eq: condition for positivity} ensures non-negative entries.
\qed
\end{assumption}

%=============================
% Subsec: Stability of Second-Order Error Moment
%=============================
\subsection{Stability of Second-Order Error Moment}
\label{subsec: stability of second order moment}
We first show that algorithm~\eqref{eq: distributed algorithm}, in the absence of gradient noise, converges and has a unique fixed-point. Then, we analyze the distance between this point and the vectors $w^o_{k,\eta}$ and $\bw_{k,i}$ in the mean-square-sense.
%=============================
% Subsubsec: Existence of fixed point
%=============================
\subsubsection{Existence and uniqueness of fixed-point}
Without gradient noise, relation~\eqref{eq: network vector recursion} reduces to:
\begin{equation}
\label{eq: network vector recursion deterministic}
\cw_i=\left(I_{MN}-\mu\eta\cL\right)\left(\cw_{i-1}-\mu\,\col\left\{\nabla_{w_k}J_k(w_{k,i-1})\right\}_{k=1}^N\right).
\end{equation}
Let $\cx\triangleq\col\{x_1,\ldots,x_N\}$ denote an $N\times1$ block vector, where $x_k$ is $M\times1$. The mapping~\eqref{eq: network vector recursion deterministic} is equivalent to the deterministic mapping $\cx\rightarrow \cy$ defined as:
\begin{equation}
\label{eq: deterministic mapping}
\cy=\left(I_{MN}-\mu\eta\cL\right)\left(\cx-\mu\,\col\left\{\nabla_{w_k}J_k(x_k)\right\}_{k=1}^N\right).
\end{equation}
\begin{lemma}{\emph{(Contractive mapping)}}
\label{lem: fixed point convergence}
Under Assumption~\ref{assumption: strong convexity} and condition~\eqref{eq: condition for stability}, the deterministic mapping defined in~\eqref{eq: deterministic mapping} satisfies:
\begin{equation}
\|\cy^1-\cy^2\|\leq\gamma\|\cx^1-\cx^2\|,
\end{equation}
with $\gamma\triangleq\max_{1\leq k\leq N}\{\gamma_k\}$ where:
\begin{equation}
\label{eq: gamma_k}
\gamma_k\triangleq\max\{|1-\mu\lambda_{k,\min}|,|1-\mu\lambda_{k,\max}|\}.
\end{equation}
This mapping is contractive when $\mu$ satisfies:
\begin{equation}
\label{eq: condition 1}
0<\mu<\min_{1\leq k\leq N}\left\{\frac{2}{\lambda_{k,\max}}\right\}.
\end{equation}
\end{lemma}
\begin{proof}
See Appendix~\ref{app: proof of convergence to a fixed point}.
\end{proof}
It then follows from Banach's fixed point theorem~\cite[pp. 299--303]{kreyszig1989introductory} that iteration~\eqref{eq: network vector recursion deterministic} converges to a unique fixed point $\cw_{\infty}=\lim_{i\rightarrow\infty}\cw_i=\col\{w_{1,\infty},\ldots,w_{N,\infty}\}$ at an exponential rate given by $\gamma$. Observe that this fixed point  is not $\cw^o_{\eta}$. Since we wish to study $\limsup_{i\rightarrow\infty}\expec\|\cw^o_{\eta}-\bcw_i\|^2$, which can be decomposed as:
\begin{align}
\limsup_{i\rightarrow\infty}\expec\|\cw^o_{\eta}-\bcw_i\|^2&=\limsup_{i\rightarrow\infty}\expec\|\cw^o_{\eta}-\cw_{\infty}+\cw_{\infty}-\bcw_i\|^2\notag\\
&\leq2\|\cw^o_{\eta}-\cw_{\infty}\|^2+2\limsup_{i\rightarrow\infty}\expec\|\cw_{\infty}-\bcw_i\|^2,\label{eq: bound on mean-square expectation}
\end{align}
we shall first asses the size of $\|\cw^o_{\eta}-\cw_{\infty}\|^2$ and then examine the quantity $\limsup_{i\rightarrow\infty}\expec\|\cw_{\infty}-\bcw_i\|^2$.

%=============================
% subsubsec: Fixed point bias Analysis
%=============================
\subsubsection{Fixed point bias analysis}
\label{eq: subsec fixed point bias analysis}
Now we analyze how far this fixed point $\cw_{\infty}$ is from the desired solution $\cw^o_{\eta}$ when the step-size $\mu$ is small. We carry out the analysis in two steps. First, we derive an expression for $\cwt_{\infty}\triangleq \cw^o_{\eta}-\cw_{\infty}$ and then we asses its size. Since $\cw_{\infty}$ is the fixed point of~\eqref{eq: network vector recursion deterministic}, we have at convergence:
\begin{equation}
\label{eq: recursion 1}
\boxed{\cw_{\infty}=\left(I_{MN}-\mu\eta\cL\right)\left(\cw_{\infty}-\mu\,\col\left\{\nabla_{w_k}J_k(w_{k,\infty})\right\}_{k=1}^N\right)}
\end{equation}
Let $\wt_{k,\infty}\triangleq w^o_{k,\eta}-w_{k,\infty}$. Using the mean-value {theorem~\cite[pp.~24]{polyak1987introduction},\cite[Appendix~D]{sayed2014adaptation}}, we can write:
\begin{equation}
\label{eq: mean value theorem for the error recursion}
\nabla_{w_k}J_k(w_{k,\infty})=\nabla_{w_k}J_k(w_{k,\eta}^o)-H_{k,\infty}\wt_{k,\infty},
\end{equation}
where
\begin{equation}
H_{k,\infty}\triangleq\int_{0}^1\nabla^2_{w_k}J_k(w^o_{k,\eta}-t\wt_{k,\infty})dt.
\end{equation}
Subtracting the vector $(I_{MN}-\mu\eta\cL)\cw^o_{\eta}$ from both sides of~\eqref{eq: recursion 1} and using relation~\eqref{eq: mean value theorem for the error recursion}, we obtain:
\begin{equation}
\label{eq: recursion 2 a}
\cwt_{\infty}=(I_{MN}-\mu\eta\cL)(I_{MN}-\mu\mathcal{H}_{\infty})\cwt_{\infty}+\mu\eta\cL \cw^o_{\eta}+\mu(I_{MN}-\mu\eta\cL)\col\left\{\nabla_{w_k}J_k(w_{k,\eta}^o)\right\}_{k=1}^N,
\end{equation}
where $\mathcal{H}_{\infty}\triangleq\diag\{H_{1,\infty},\ldots,H_{N,\infty}\}$. From~\eqref{eq: optimality condition}, recursion~\eqref{eq: recursion 2 a} can be written alternatively as:
\begin{equation}
\label{eq: recursion 2 b}
\cwt_{\infty}=(I_{MN}-\mu\eta\cL)(I_{MN}-\mu\mathcal{H}_{\infty})\cwt_{\infty}+\mu^2\eta^2\cL^2\cw^o_\eta,
\end{equation}
so that:
\begin{equation}
\label{eq: recursion 2}
\boxed{\cwt_{\infty}=\mu^2\eta^2\left[I_{MN}-(I_{MN}-\mu\eta\cL)(I_{MN}-\mu\mathcal{H}_{\infty})\right]^{-1}\cL^2\cw^o_\eta}
\end{equation}
The inverse exists when $(I_{MN}-\mu\eta\cL)(I_{MN}-\mu\mathcal{H}_{\infty})$ is stable, i.e., its spectral radius is less than one. Since the spectral radius of a matrix is upper bounded by any of its induced {norms}, we have:
\begin{equation}
\rho((I_{MN}-\mu\eta\cL)(I_{MN}-\mu\mathcal{H}_{\infty}))\leq\|I_{MN}-\mu\eta\cL\|\|I_{MN}-\mu\mathcal{H}_{\infty}\|,
\end{equation}
in terms of the $2-$induced norm. Under condition~\eqref{eq: condition for stability} and since $\lambda_1(L)=0$, we have $\|I_{MN}-\mu\eta\cL\|=1$. From Assumption~\ref{assumption: strong convexity}, we have:
\begin{equation}
(1-\mu\lambda_{k,\max})I_M\leq I_M-\mu H_{k,\infty}\leq(1-\mu\lambda_{k,\min})I_M,
\end{equation}
so that $\|I_{MN}-\mu\mathcal{H}_{\infty}\|_2\leq \max_{1\leq k\leq N}\gamma_k$ with $\gamma_k$ given in~\eqref{eq: gamma_k}. We conclude that when~\eqref{eq: condition for stability}  and~\eqref{eq: condition 1} are satisfied, the inverse exists.

From~\eqref{eq: recursion 2}, we observe that $\cwt_{\infty}$ is zero {in two cases: i) when $\eta=0$; ii) when $w^o_k=w^o$ $\forall k$, i.e., $\cw^o=\mathds{1}_N\otimes w^o$. In the second case, consider~\eqref{eq: GFT of w o eta 2} and observe that $\wb^o_1=\sqrt{N}w^o$, $\left[\cwb^o\right]_{2:N}=0$, and $\cwb^o_{\eta}=\col\{\wb^o_1,0\}$. Thus, $\cw^o_{\eta}=(V\otimes I_M)\cwb^o_{\eta}=(v_1\otimes I_M)\wb^o_1=\mathds{1}_N\otimes w^o$ and $\cL\cw^o_\eta=\cL(\mathds{1}_N\otimes w^o)=0$.}

\begin{theorem}{\emph{(Fixed point bias size)}} \label{theo: dimension of the bias}
Under Assumption~\ref{assumption: strong convexity} and for small $\mu$ {satisfying conditions~\eqref{eq: condition for stability} and~\eqref{eq: condition 1}}, the steady-state bias $\cwt_{\infty}=\cw^o_{\eta}-\cw_{\infty}$ of the mapping~\eqref{eq: network vector recursion deterministic} satisfies:
\begin{equation}
\label{eq: bound on the bias}
{\mu\lim_{\mu\rightarrow 0}\left(\frac{1}{\mu}\|\cw^o_{\eta}-\cw_{\infty}\|\right)\leq \mu\frac{O(\eta^2)}{(O(1)+O(\eta))^2}.}
\end{equation}
\end{theorem}
\begin{proof}
See Appendix~\ref{app: proof of dimension of the bias}.
\end{proof}

%=============================
% subsubsec: Evolution of stochastic recursion
%=============================
\subsubsection{Evolution of the stochastic recursion}
We now examine how close the stochastic algorithm~\eqref{eq: distributed algorithm} approaches $\cw^o_{\eta}$. First, we introduce the mean-square perturbation vector (MSP) at time $i$ relative to $\cw_{\infty}$:
\begin{equation}
\text{MSP}_i\triangleq\col\left\{\expec\|w_{k,\infty}-\bw_{k,i}\|^2\right\}_{k=1}^N.
\end{equation}
The $k$-th entry of $\text{MSP}_i$ characterizes how far away the estimate $\bw_{k,i}$ at agent $k$ and time $i$ is from $w_{k,\infty}$.
\begin{theorem}{\emph{(Network mean-square-error stability)}}
\label{theo: dimension of the MSP}
Under Assumptions~\ref{assumption: strong convexity},~\ref{assumption: gradient noise}, and~\ref{assumption: combination matrix}, the \emph{MSP} at time $i$ can be recursively bounded as:
\begin{equation}
\label{eq: evolution of the MSP i}
\emph{MSP}_i \preceq (I_N-\mu\eta L)\,G\,\emph{MSP}_{i-1}+\mu^2 (I_N-\mu\eta L)b,
\end{equation}
where:
\begin{align}
G&\triangleq{\emph\diag}\left\{\gamma_k^2+3\mu^2\beta^2_k\right\}_{k=1}^N,\\
\label{eq: equation of b}b&\triangleq{\emph\col}\left\{\sigma^2_{s,k}+3\beta_k^2\|w^o_{k,\eta}\|^2+ 3\beta_k^2\|w^o_{k,\eta}-w_{k,\infty}\|^2\right\}_{k=1}^N.
%&=O(1)+O(\mu^2\eta^4)(O(1)+O(\eta))^{-4}.\nonumber
\end{align}
A sufficient condition for the stability of the above recursion is:
\begin{equation}
\label{eq: MSP stability}
0<\mu<\min_{1\leq k\leq N}{\left\{\min\left\{\frac{2\lambda_{k,\min}}{\lambda_{k,\min}^2+3\beta_k^2},\frac{2\lambda_{k,\max}}{\lambda_{k,\max}^2+3\beta_k^2}\right\}\right\}}.
\end{equation}
It follows that
\begin{equation}
\label{eq: steady-state of the MSP}
\|\limsup_{i\rightarrow\infty}{\emph{MSP}}_i\|_{\infty}=O(\mu),
\end{equation}
and% the steady-state network $\emph{MSD}\triangleq\displaystyle\limsup_{i\rightarrow\infty}\frac{1}{N}\expec\|\cw^o_{\eta}-\bcw_i\|^2$ is:
\begin{equation}
\label{eq: steady-state of the second order}
\limsup_{i\rightarrow\infty}\expec\|\cw^o_{\eta}-\bcw_i\|^2=O(\mu)+\frac{O(\mu^2\eta^4)}{(O(1)+O(\eta))^{4}}= O(\mu).
\end{equation}
\end{theorem}
\begin{proof}
See Appendix~\ref{app: proof of dimension of the MSP}. {With regards to~\eqref{eq: steady-state of the second order} note first that for fixed $\eta$, we have $O(\mu)+O(\mu^2)=O(\mu)$. When $\eta$ and $\mu$ are coupled ($\eta=\mu^{-\epsilon}$), we obtain:
$$O(\mu)+\frac{O(\mu^{2-4\epsilon})}{O(1)+O(\mu^{-4\epsilon})}.$$
For $\epsilon<0$, $O(1)$ dominates $O(\mu^{-4\epsilon})$ in the denominator and we obtain $O(\mu)+O(\mu^{2-4\epsilon})=O(\mu)$. For $\epsilon>0$, $O(\mu^{-4\epsilon})$ dominates $O(1)$ in the denominator and we obtain $O(\mu)+O(\mu^2)=O(\mu)$.}

\end{proof}

%======================================
% Subsec: Stability of Fourth-Order Error Moment
%======================================
\subsection{Stability of Fourth-Order Error Moment}
{The results so far establish that the iterates $\bw_{k,i}$ converge to a small $O(\mu)-$ neighborhood around the regularized solution $w^o_{k,\eta}$. We can be more precise and determine the size of this neighborhood, i.e., assess the size of the constant multiplying $\mu$ in the $O(\mu)-$term. To do so, we shall derive {in Part II~\cite{nassif2018diffusion}}  an accurate first-order expression for the mean-square error~\eqref{eq: steady-state of the second order}; the expression will be accurate to first-order in $\mu$. This expression will be useful because it will allow us to highlight several features of the limiting point of the network as a function of the parameter $\eta$.}

{To arrive at the desired expression, we first need to introduce a long-term approximation model and assess how close it is to the actual model. We then derive the performance for the long-term model and use this closeness to transform this result into an accurate expression for the performance of the original learning algorithm. When this argument is concluded we arrive at the desired performance expression, which we then use to comment on the behavior of the algorithm in a more informed manner. To derive the long-term model, we shall follow the approach developed in~\cite{sayed2014adaptation}. The first step is to establish the asymptotic stability of the}  fourth-order moment of the error vector, $\expec\|\cw_{\eta}^o-\bcw_{i}\|^4$.  {This property is needed to justify the validity of the long-term approximate model {that will be introduced in  Part II~\cite{nassif2018diffusion}}.}

To establish the fourth-order stability, we {replace condition~\eqref{eq: condition on second-order moment of gradient noise} on the gradient noise process by the following condition on its fourth order moment:}
\begin{equation}
\label{eq: condition on fourth-order moment of gradient noise}
\expec\left[\|\bs_{k,i}(\bw_k)\|^4|\bcF_{i-1}\right]\leq\overline{\beta}_k^4\|\bw_k\|^4+\overline{\sigma}_{s,k}^4,
\end{equation}
for some $\overline{\beta}_k^4\geq 0$,  and $\overline{\sigma}_{s,k}^4\geq 0$. As explained in~\cite{sayed2014adaptation}, condition~\eqref{eq: condition on fourth-order moment of gradient noise} {implies~\eqref{eq: condition on second-order moment of gradient noise} and, likewise, condition~\eqref{eq: condition on fourth-order moment of gradient noise} holds for important cases of interest.}

Exploiting the convexity of the norm functions $\|x\|^4$ and $\|x\|^2$ and using Jensen's inequality, we can write:
\begin{equation}
\label{eq: fourth order bound 1}
\expec\|\cw_{\eta}^o-\bcw_{i}\|^4\leq 8\|\cw_{\eta}^o-\cw_{\infty}\|^4+8\expec\|\cw_{\infty}-\bcw_{i}\|^4,
\end{equation}
and
\begin{equation}
\label{eq: fourth order bound 2}
\begin{split}
\expec\|\cw_{\infty}-\bcw_{i}\|^4=\expec\left(\|\cw_{\infty}-\bcw_{i}\|^2\right)^2=\expec\left(\sum_{k=1}^N\|w_{k,\infty}-\bw_{k,i}\|^2\right)^2&=N^2\expec\left(\sum_{k=1}^N\frac{1}{N}\|w_{k,\infty}-\bw_{k,i}\|^2\right)^2\\
&\leq N\sum_{k=1}^N\expec\|w_{k,\infty}-\bw_{k,i}\|^4.
\end{split}
\end{equation}
Let us introduce the mean-fourth perturbation vector at time $i$ relative to $\cw_{\infty}$:
\begin{equation}
\text{MFP}_i\triangleq\col\left\{\expec\|w_{k,\infty}-\bw_{k,i}\|^4\right\}_{k=1}^N.
\end{equation}

\begin{theorem}{\emph{(Fourth-order error moment stability)}}
\label{theo: dimension of the MFP}
Under Assumptions~\ref{assumption: strong convexity},~\ref{assumption: gradient noise},~\ref{assumption: combination matrix}, and condition~\eqref{eq: condition on fourth-order moment of gradient noise}, the \emph{MFP} at time $i$ can be recursively bounded as:
\begin{equation}
\label{eq: evolution of the MFP i}
\emph{MFP}_i\preceq (I_N-\mu\eta L)G'\emph{MFP}_{i-1}+\mu^2(I_N-\mu\eta L)B\emph{MSP}_{i-1}+\mu^4(I_N-\mu\eta L)b',
\end{equation}
where
\begin{align}
G'&\triangleq\emph{\diag}\left\{\gamma_k^4+24\mu^2\gamma_k^2\beta_k^2+81\mu^4\overline{\beta}_k^4\right\}_{k=1}^N,\\
B&\triangleq8\gamma_k^2\emph{\diag}\left\{\sigma^2_{s,k}+3\beta_k^2\|w^o_{k,\eta}\|^2+3\beta_k^2\|w^o_{k,\eta}-w_{k,\infty}\|^2\right\}_{k=1}^N,\\
b'&\triangleq\emph{\col}\left\{3\overline{\sigma}^4_{s,k}+81\overline{\beta}_k^4\|w^o_{k,\eta}\|^4+81\overline{\beta}_k^4\|w^o_{k,\eta}-w_{k,\infty}\|^4\right\}_{k=1}^N.
\end{align}
A sufficiently small $\mu$ ensures the stability of the above recursion. It follows that
\begin{equation}
\label{eq: steady-state of the MFP i}
\|\limsup_{i\rightarrow\infty}{\emph{MFP}}_i\|_{\infty}=O(\mu^2),
\end{equation}
and
\begin{equation}
\label{eq: steady-state of the fourth order}
\limsup_{i\rightarrow\infty}\expec\|\cw^o_{\eta}-\bcw_i\|^4=O(\mu^2)+\frac{O(\mu^4\eta^8)}{(O(1)+O(\eta))^{8}}= O(\mu^2).
\end{equation}
\end{theorem}
\begin{proof}
See Appendix~\ref{app: proof of dimension of the MFP}.
\end{proof}

%======================================
% Subsec: Stability of First-order Moment
%======================================
\subsection{Stability of First-order Error Moment}
%Using the fact that $(\expec\ba)^2\leq\expec\ba^2$ for any real-valued random variable $\ba$, we can readily conclude that:
%\begin{equation}
%\limsup_{i\rightarrow\infty}\expec\|\bwt_{k,i}\|=O(\mu^{1/2}),\qquad k=1,\ldots, N
%\end{equation}
%so that the first-order moment of the error vector tends to a bounded region of $O(\mu^{1/2})$. However, a smaller bound on $\|\expec\bwt_{k,i}\|$ can be derived with $O(\mu^{1/2})$ replaced by $O(\mu)$.
{We next need to} examine the evolution of the mean-error vector $\expec(\cw_{\eta}^o-\bcw_{i})$. To establish the mean-stability, we need to introduce a smoothness condition on the Hessian matrices of the individual costs. This smoothness condition will be adopted in the next {Part II~\cite{nassif2018diffusion}} when we study the long term behavior of the network.
\begin{assumption}{\emph{(Smoothness condition on individual cost functions).}}
\label{assumption: local smoothness hessian}
It is assumed that each $J_k(w_k)$ satisfies a smoothness condition close to $w^o_{k,\eta}$, in that the corresponding Hessian matrix is Lipchitz continuous in the proximity of $w^o_{k,\eta}$ with some parameter $\kappa_d\geq 0$, i.e.,
\begin{equation}
\label{eq: hessian smoothness condition}
\|\nabla^2_{w_k}J_k(w^o_{k,\eta}+\Delta w_k)-\nabla^2_{w_k}J_k(w^o_{k,\eta})\|\leq \kappa_d\|\Delta w_k\|,
\end{equation}
for small perturbations $\|\Delta w_k\|\leq\epsilon$.
\qed
\end{assumption}
\noindent From the triangle inequality, we have:
\begin{equation}
\label{eq: triangle inequality}
\|\expec(\cw_{\eta}^o-\bcw_{i})\|\leq \|\cw_{\eta}^o-\cw_{\infty}\| + \|\expec(\cw_{\infty}-\bcw_{i})\|.
\end{equation}
Let us introduce the square-mean perturbation (SMP) vector at time $i$ relative to $\cw_{\infty}$:
\begin{equation}
\label{eq: steady-state of the MSP i}
\text{SMP}_i\triangleq\col\left\{\|\expec (w_{k,\infty}-\bw_{k,i})\|^2\right\}_{k=1}^N.
\end{equation}

\begin{theorem}{\emph{(First-order error moment stability)}}
\label{theo: dimension of the SMP}
Under Assumptions~\ref{assumption: strong convexity},~\ref{assumption: gradient noise},~\ref{assumption: combination matrix}, and~\ref{assumption: local smoothness hessian}, the \emph{SMP} at time $i$ can be recursively bounded as:
\begin{equation}
\label{eq: evolution of the SMP i}
\emph{SMP}_i\preceq (I_N-\mu\eta L)G''\emph{SMP}_{i-1}+\mu^2(I_N-\mu\eta L)(I_N-G'')^{-1}B'\emph{MSP}_{i-1}+\mu^2\frac{1}{2}(\kappa'_d)^2(I_N-\mu\eta L)(I_N-G'')^{-1}\emph{MFP}_{i-1}.
\end{equation}
where
\begin{align}
G''&\triangleq\emph{\diag}\left\{\gamma_k\right\}_{k=1}^N,\label{eq: equation of G''}\\
B'&\triangleq2(\kappa'_d)^2\emph{\diag}\left\{\|w^o_{k,\eta}-w_{k,\infty}\|^2\right\}_{k=1}^N,
\end{align}
with $\kappa'_d=\max\{\kappa_d,\frac{\lambda_{k,\max}-\lambda_{k,\min}}{\epsilon}\}$. Under condition~\eqref{eq: condition on fourth-order moment of gradient noise}, a sufficiently small $\mu$ ensures the stability of the above recursion. It follows that
\begin{equation}
\label{eq: steady-state of the SMP i}
\|\limsup_{i\rightarrow\infty}{\emph{SMP}}_i\|_{\infty}=O(\mu^2),
\end{equation}
and that
\begin{equation}
\label{eq: steady-state of the first order}
\limsup_{i\rightarrow\infty}\|\expec(\cw^o_{\eta}-\bcw_i)\|=O(\mu)+\frac{O(\mu\eta^2)}{(O(1)+O(\eta))^{2}}.
\end{equation}
\end{theorem}
\begin{proof}
See Appendix~\ref{app: proof of dimension of the SMP}.
\end{proof}

{We have established so far the stability of the mean-error process, $\expec(\cw^o_\eta-\bcw_i)$, the mean-square-error $\expec\|\cw^o_\eta-\bcw_i\|^2$, and the fourth order moment $\expec\|\cw^o_\eta-\bcw_i\|^4$. Building on these results, we will derive in Part II~\cite{nassif2018diffusion} closed form expressions for the steady-state performance of algorithm~\eqref{eq: distributed algorithm}. Section VI in Part II~\cite{nassif2018diffusion} will provide illustration for the theoretical results in this part (Theorems~\ref{theo: dimension of the bias},~\ref{theo: dimension of the MSP}, and~\ref{theo: dimension of the SMP}), and its accompanying Part II. }

%======================================
% Sec: Simulation results with real dataset
%======================================
\section{Simulation results with real dataset}
{In this section, we test algorithm~\eqref{eq: distributed algorithm} on a weather dataset corresponding to a collection of daily measurements (mean temperature, mean dew point, mean visibility, mean wind speed, maximum sustained wind speed, and rain or snow occurrence) taken from 2004 to 2017 at $N=139$ weather stations located around the continental United States~\cite{lawrimorenoaa}. We construct a representation graph $\cG=(\cN,\cE,A)$ for the stations using geographical distances between sensors. Each sensor corresponds to a node $k$ and is connected to $|\cN_k|$ neighbor nodes with undirected edges weighted according to $a_{k\ell}=\frac{1}{2}(p_{k\ell}+p_{\ell k})$ with~\cite{sandryhaila2013discrete}:
\begin{equation}
p_{k\ell}=\frac{e^{-d^2_{k\ell}}}{\sqrt{\sum_{m\in\cN_{k,0}}e^{-d^2_{k m}}\sum_{n\in\cN_{\ell,0}}e^{-d^2_{\ell n}}}},\quad\ell\in\cN_{k,0},
\end{equation}
where $\cN_{k,0}$ is the set of $4$-nearest neighbors of node $k$ and $d_{k\ell}$ denotes the geodesical distance between the $k$-th and  $\ell$-sensors -- see Fig.~\ref{fig: prediction rain} (left). Let $h_{k,i}\in\mathbb{R}^M$ denote the feature vector at sensor $k$ and day $i$ composed of $M=5$ entries corresponding to the mean temperature, mean dew point, mean visibility, mean wind speed, and maximum sustained wind speed reported at day $i$ at sensor $k$. Let $\gamma_{k}(i)$ denote a binary variable associated with the occurrence of rain (or snow) at node $k$ and day $i$, i.e, $\gamma_{k}(i)=1$ if rain (or snow) occurred and $\gamma_{k}(i)=-1$ otherwise. We would like to construct a classifier that allows us to predict whether it will rain (or snow) or not based on the knowledge of the feature vector $h_{k,i}$. In principle, each station could use an individual logistic regression machine~\cite{sayed2014adaptation,hosmer2013applied,theodoridis2008pattern}, that seeks a vector  $w^o_k$, such that $\widehat{\boldsymbol{\gamma}}_{k}(i)=\sign(\boldsymbol{h}_{k,i}^\top w^o_k)$ and 
\begin{equation}
w^o_k\triangleq\arg\min_{w_k}\expec\ln\left(1+e^{-\boldsymbol{\gamma}_k(i)\boldsymbol{h}_{k,i}^{\top}w_k}\right)+\rho\|w_k\|^2.
\end{equation}
In this application, however, it is expected that the decision rules $\{w^o_k\}$ at neighboring stations will be similar. In the experiment, the dataset is split into a training set used to learn the decision rule $w^o_k$, and a test set from which $\widehat{\boldsymbol{\gamma}}_{k}(i)$ are generated for performance evaluation. The first dataset comprises daily weather data recorded at the stations in the interval $2004-2012$ (a total number of $D_a=3288$ days) and the training set contains data recorded in the interval $2012-2017$ (a total number of $D_t=1826$ days). We set $\mu=3\cdot10^{-4}$ and $\rho=10^{-5}$. We generate the first iterate $w_{k,0}$ from the Gaussian distribution $\cN(0, I_M)$ and we run strategy~\eqref{eq: distributed algorithm} over the training set ($i=1,\ldots,D_a$) for different values of $\eta$. For each value of $\eta$, we report in Table~\ref{table: prediction error} the prediction error over the test set defined as: 
\begin{equation}
\label{eq: prediction error}
%\frac{1}{N}\sum_{k=1}^{N}\frac{1}{D_t}\sum_{i=1}^{D_t=1826}\frac{|\gamma_{k}(i)-\sign(\boldsymbol{h}_{k,i}^\top \widehat{w}_{k,\infty})|}{2},
\frac{1}{N}\sum_{k=1}^{N}\frac{1}{D_t}\sum_{i=1}^{D_t=1826}\mathbb{I}[\sign(\boldsymbol{h}_{k,i}^\top \widehat{w}_{k,\infty})\neq\gamma_{k}(i)],
\end{equation}
where  $N=139$ is the number of nodes, $\widehat{w}_{k,\infty}$ is the average of the last 200 iterates generated by the algorithm at agent $k$, and $\mathbb{I}[x]$ is the indicator function at $x$, namely, $\mathbb{I}[x]=1$ if $x$ is true and 0 otherwise. % value of the estimates $\bw_{k,i}$ generated by the algorithm over the last $200$ days. 
Table~\ref{table: prediction error} shows that through cooperation, the agents improve performance. This is due to the fact that the non-cooperative solution ($\eta=0$) may suffer from a slow convergence rate~\cite[Section V-B]{nassif2018distributed} in which case  some nodes may not be able to converge in the finite dataset scenario. By increasing $\eta$, the convergence rate improves. However, a large value of $\eta$ (such as $\eta=\mu^{-1}$) yields a deterioration in the accuracy since in this case all agents converge approximately to the same classifier. By setting $\eta=45$, we obtain the smallest prediction error. We show in Fig.~\ref{fig: prediction rain} (right) the results of the prediction on July 30, 2015 across the US for $\eta=45$.
\begin{figure*}
\centering
\includegraphics[scale=0.45]{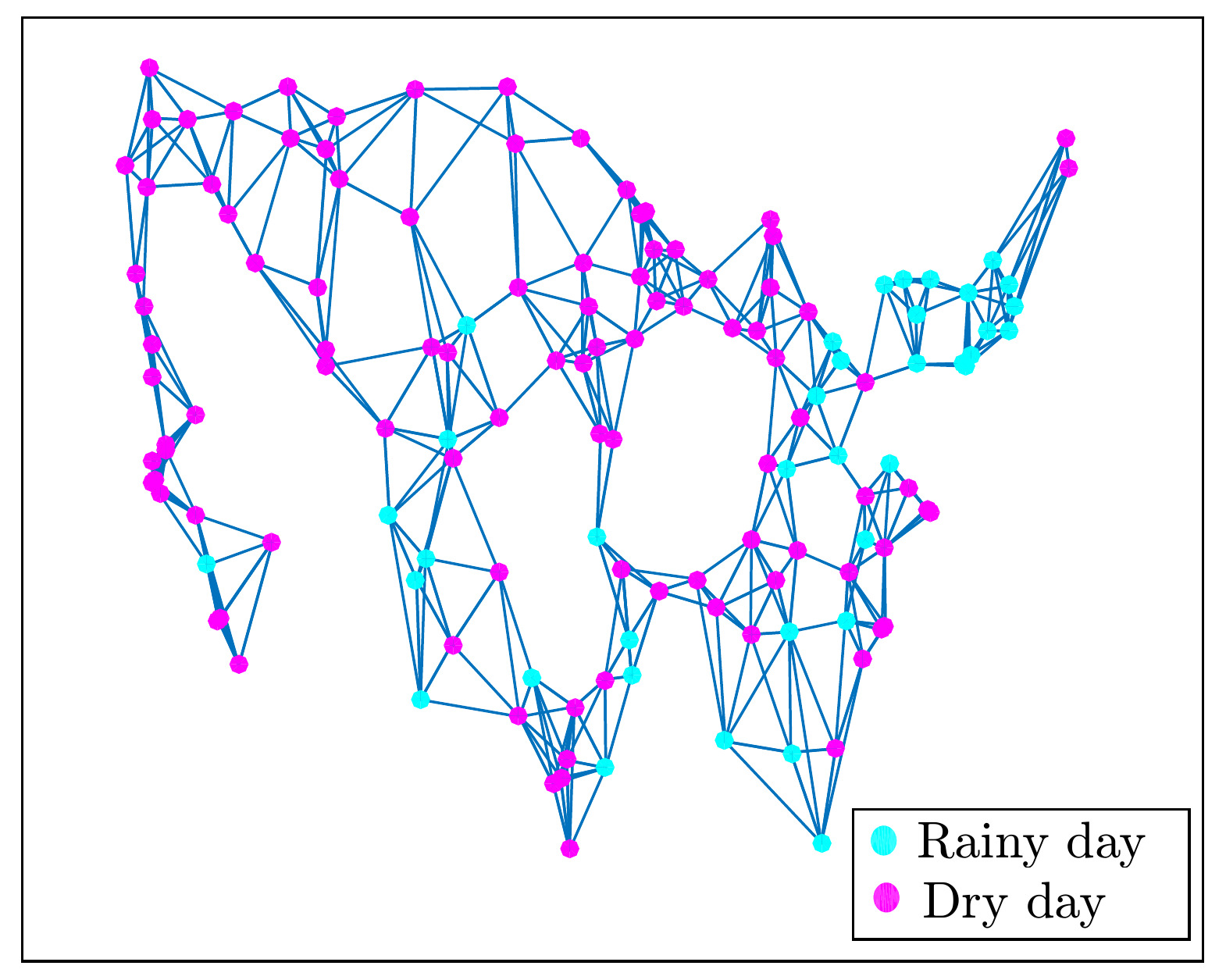}\hspace{1 cm}
\includegraphics[scale=0.45]{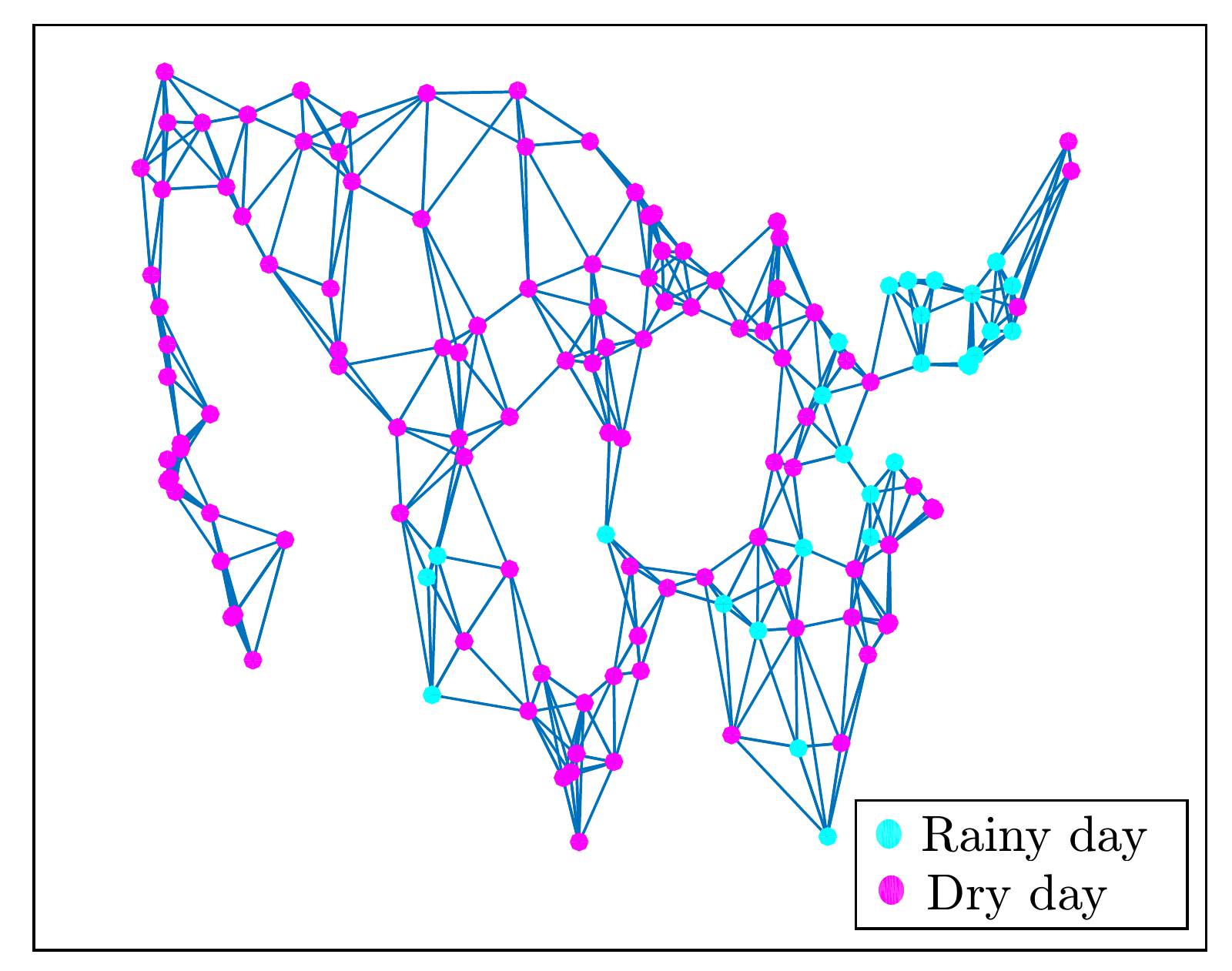}
\caption{\textit{(Left)} Occurrence of rain reported by 139 weather stations across the US on July 30, 2015. \textit{(Right)} Prediction of rain occurrence from weather data based on logistic regression and multitask learning.}
\label{fig: prediction rain}
\end{figure*}
\begin{table}
\caption{Rain prediction error~\eqref{eq: prediction error} in weather sensor networks for different values of regularization strength $\eta$.}
\begin{center}
 \begin{tabular}{c c c c c c c} 
 \hline\hline
 &$\eta=0$ & $\eta=10$ & $\eta=45$ & $\eta=100$ &$\eta=1000$&$\eta=\mu^{-1}$\\ [0.5ex] 
 \hline
prediction error &0.309 & 0.232 & \textbf{0.225} & 0.226&0.228&0.232 \\ [0.5ex] 
 \hline\hline
\end{tabular}
\end{center}
\label{table: prediction error}
\end{table}}

%======================================
% Sec: Conclusion
%======================================
\section{Conclusion}
In this work, we considered multitask inference problems where agents in the network have individual parameter vectors to estimate subject to a smoothness condition over the graph. Based on diffusion adaptation, we proposed a strategy that allows the network to minimize a global cost consisting of the aggregate sum of the individual costs regularized by a term promoting smoothness. We showed that, for small step-size parameter, the network is able to approach the minimizer of the regularized problem to arbitrarily good accuracy levels. Furthermore, we showed how the regularization strength {can steer the convergence point of the network toward many modes starting from the non-cooperative mode and ending with the single-task mode. }%Building on the stability analysis provided in this Part I, we shall characterize in the accompanying Part II~\cite{nassif2018diffusion} how close the network converges to the limiting point.}%influences the limit point and the steady-state mean-square-error (MSE) performance of the algorithm. Analytical expressions illustrating these effects are derived. %These expressions revealed explicitly the influence of the network topology, data settings, step-size parameter, and regularization strength on the network MSE performance and provided insights into the design of effective multitask strategies for distributed inference over networks.  Illustrative examples were considered and links to spectral graph filtering were also provided.

%===================================
% Appendices
%===================================
\begin{appendices}

%===================================
% Sec: Proof of Lemma related to w o eta
%===================================
\section{Proof of Lemma~\ref{lemm: transformed w o eta}}
\label{sec: Proof of Lemma transformed w o eta}
\noindent Consider the matrix inversion identity~\cite{kailath1980linear}:
\begin{equation}
\label{eq: matrix inversion identity}
(A+BCD)^{-1}=A^{-1}-A^{-1}B(C^{-1}+DA^{-1}B)^{-1}DA^{-1},
\end{equation}
which allows us to write:
\begin{equation}
\label{eq: matrix inversion U+V}
(U+W)^{-1}U=I-U^{-1}(I+WU^{-1})^{-1}W=I-(U+W)^{-1}W,
\end{equation}
for any invertible matrix $U$. Using~\eqref{eq: matrix inversion U+V}, we  write~\eqref{eq: GFT of w o eta} alternatively as:
\begin{equation}
\label{eq: GFT of w o eta 1}
\cwb^o_{\eta}=\left(I-\eta\left(\cHb_{\eta}^o+\eta\cJ\right)^{-1}\cJ\right)\cwb^o.
\end{equation}
Let
\begin{equation}
\cQ\triangleq\cHb^o_{\eta}+\eta\cJ.
\end{equation}
Using the definitions~\eqref{eq: definition cJ} and~\eqref{eq: definition cH o eta}, we can partition $\cQ$ into blocks:
\begin{equation}
\cQ=\left[\begin{array}{cc}
\cQ_{11}&\cQ_{12}\\
\cQ_{12}^\top&\cQ_{22}
\end{array}\right],
\end{equation}
with $Q_{11}$, $\cQ_{12}$, and $\cQ_{22}$ defined in~\eqref{eq: Q 11 lemma},~\eqref{eq: Q 12 lemma}, and~\eqref{eq: Q 22 lemma}, respectively.
%\begin{eqnarray}
%\cQ_{11}&\triangleq&(v_1^\top\otimes I_M)\cH^o_{\eta}(v_1\otimes I_M),\label{eq: Q11}\\
%\cQ_{12}&\triangleq&(v_1^\top\otimes I_M)\cH^o_{\eta}(V_R\otimes I_M),\label{eq: Q12}\\
%\cQ_{22}&\triangleq&(V_R^\top\otimes I_M)\cH^o_{\eta}(V_R\otimes I_M)+\eta\Lambda_o\otimes I_M.\label{eq: Q22}
%\end{eqnarray}
Since $v_1=\frac{1}{\sqrt{N}}\mathds{1}_N$, we have $\cQ_{11}=\frac{1}{N}\sum_{k=1}^NH^o_{k,\eta}$ which is {positive definite} from Assumption~\ref{assumption: strong convexity}. Observe that $\cQ$ is invertible since it is similar to $\cH^o_{\eta}+\eta\cL$ which is positive definite under Assumption~\ref{assumption: strong convexity}. Now, by applying the block inversion formula to $\cQ$, we obtain:
\begin{align}
\label{eq: inverse of Q}
\cQ^{-1}&=(\cHb^o_{\eta}+\eta\cJ)^{-1}\nonumber\\
&=\left[\begin{array}{cc}
\cQ_{11}^{-1}+\cQ_{11}^{-1}\cQ_{12}\cG\cQ_{12}^{\top}\cQ_{11}^{-1}&-\cQ_{11}^{-1}\cQ_{12}\cG\\
-\cG\cQ_{12}^{\top}\cQ_{11}^{-1}&\cG
\end{array}\right],
\end{align}
where
\begin{equation}
%\label{eq: G11}
%\cG_{11}&\triangleq&\cQ_{11}^{-1}+\cQ_{11}^{-1}\cQ_{12}\cG_{22}\cQ_{21}\cQ_{11}^{-1},\\
%\label{eq: G 12}
%\cG_{12}&\triangleq&-\cQ_{11}^{-1}\cQ_{12}\cG_{22},\\
%\label{eq: G21}
%\cG_{21}&\triangleq&-\cG_{22}\cQ_{21}\cQ_{11}^{-1},\\
\label{eq: G 22}\cG\triangleq(\cQ_{22}-\cQ_{12}^{\top}\cQ_{11}^{-1}\cQ_{12})^{-1}.
\end{equation}
Replacing~\eqref{eq: inverse of Q} into~\eqref{eq: GFT of w o eta 1} and using~\eqref{eq: definition cJ}, we arrive at:
\begin{equation}
\label{eq: c w o eta bar}
\cwb^o_{\eta}=\left[
\begin{array}{cc}
I_M&\eta\cQ_{11}^{-1}\cQ_{12}\cG\left(\Lambda_o\otimes I_M\right)\\
0& I_{M(N-1)}-\eta\,\cG\left(\Lambda_o\otimes I_M\right)
\end{array}
\right]\cwb^o.
\end{equation}
Using definition~\eqref{eq: K 22 lemma 1} into~\eqref{eq: c w o eta bar}, we conclude~\eqref{eq: GFT of w o eta 2}.
% Let $\cK\triangleq\left(I-\eta\left(\cHb_{\eta}^o+\eta\cJ\right)^{-1}\cJ\right)$, which can be partitioned into blocks:
%\begin{align}
%\label{eq: K11}
%\cK_{11}&\triangleq I_M,\\
%\label{eq: K12}
%\cK_{12}&\triangleq-\eta\,\cG_{12}\left(\Lambda_o\otimes I_M\right),\\
%\label{eq: K21}
%\cK_{21}&\triangleq0,\\
%\label{eq: K22}
%\cK_{22}&\triangleq I_{M(N-1)}-\eta\,\cG_{22}\left(\Lambda_o\otimes I_M\right).
%\end{align}
%From~\eqref{eq: K12},~\eqref{eq: G 12}, and~\eqref{eq: K22}, we have:
%\begin{align}
%\cK_{12}&=\cQ_{11}^{-1}\cQ_{12}\left(I_{M(N-1)}-\cK_{22}\right)\label{eq: K12 1}.
%\end{align}
%Thus, from~\eqref{eq: GFT of w o eta 1}, and~\eqref{eq: K11}--\eqref{eq: K12 1}, we conclude~\eqref{eq: GFT of w o eta 2}.
%
%\cK&\triangleq I_{M(N-1)}-\eta\,\cG\left(\Lambda_o\otimes I_M\right)\label{eq: K 22 lemma 1}\\
%&=\cG((V_R^\top\otimes I_M)\cH^o_{\eta}(V_R\otimes I_M)-\cQ_{12}^\top\cQ_{11}^{-1}\cQ_{12})\label{eq: K 22 lemma}

Now, we establish~\eqref{eq: order of K 22}. Let us first introduce the matrix $\cG'$:
\begin{equation}
\label{eq: cG'}
\cG'\triangleq(V_R^\top\otimes I_M)\cH_{\eta}^o(V_R\otimes I_M)-\cQ_{12}^{\top}\cQ_{11}^{-1}\cQ_{12}.
\end{equation}
Using the above definition {and expressions~\eqref{eq: G 22 lemma} and~\eqref{eq: Q 22 lemma}}, we can re-write the matrix $\cK$ in~\eqref{eq: K 22 lemma 1}  alternatively as:
\begin{align}
\cK&={I_{M(N-1)}-\eta\,((V_R^\top\otimes I_M)\cH^o_{\eta}(V_R\otimes I_M)+\eta\Lambda_o\otimes I_M-\cQ_{12}^\top\cQ_{11}^{-1}\cQ_{12})^{-1}\left(\Lambda_o\otimes I_M\right)}\notag\\
&={I_{M(N-1)}-\eta(\cG'+\eta \Lambda_o\otimes I_M)^{-1}(\Lambda_o\otimes I_M)}\notag\\
&\overset{\eqref{eq: matrix inversion U+V}}{=}(\cG'+\eta \Lambda_o\otimes I_M)^{-1}\cG'.	\label{eq: alternative K 22}
\end{align}
{The matrix $\cG'$ in~\eqref{eq: cG'} is the Schur complement of $\cHb^o_{\eta}$ in~\eqref{eq: definition cH o eta} which can be partitioned as:
\begin{equation}
\cHb^o_{\eta}=\left[\begin{array}{cc}
\cQ_{11}&\cQ_{12}\\
\cQ_{12}^\top&(V_R^\top\otimes I_M)\cH^o_{\eta}(V_R\otimes I_M)
\end{array}\right].
\end{equation}
Thus, $\cG'$} is positive definite since it is the Schur complement of the positive definite matrix $\cHb^o_{\eta}$~\cite[pp.~651]{boyd2004convex}. %Consider the symmetric matrix $\cG_{22}$ in~\eqref{eq: G 22}, whose inverse can be written as:
%\begin{equation}
%\label{eq: argument 1 G22}
%\cG_{22}^{-1}=\cG'+\eta \Lambda_o\otimes I_M,
%\end{equation}
Since $\cG'$ is symmetric, from Weyl's inequality~\cite[pp.~239]{horn2003matrix} we have:
\begin{equation}
\label{eq: argument 2 G22}
0<\eta\lambda_2(L)+\lambda_{\min}(\cG')\leq\lambda_{\min}(\cG'+\eta \Lambda_o\otimes I_M)\leq\eta\lambda_2(L)+\lambda_{\max}(\cG').
\end{equation}
Furthermore, since {$\cG'$ is the Schur complement of the positive definite matrix $\cHb^o_{\eta}$}, we {have~\cite[Theorem~5]{smith1992some}}:
\begin{equation}
\label{eq: bound relation a 4}
\lambda_{\min}(\cG')\geq\lambda_{\min}(\cHb^o_{\eta})=\lambda_{\min}(\cH^o_{\eta})\geq\min_{1\leq k \leq N}\lambda_{k,\min},
\end{equation}
\begin{equation}
\label{eq: bound relation a 3}
\lambda_{\max}(\cG')\leq\lambda_{\max}(\cHb^o_{\eta})=\lambda_{\max}(\cH^o_{\eta})\leq\max_{1\leq k \leq N}\lambda_{k,\max}.
\end{equation}
Therefore, from~\eqref{eq: argument 2 G22} and~\eqref{eq: bound relation a 4}, we get:
\begin{equation}
\label{eq: argument 3 G22}
\lambda_{\min}(\cG'+\eta \Lambda_o\otimes I_M)\geq \eta\lambda_2(L)+\min_{1\leq k \leq N}\lambda_{k,\min},
\end{equation}
and
\begin{equation}
\label{eq: argument 4 G22}
\lambda_{\max}((\cG'+\eta \Lambda_o\otimes I_M)^{-1})=\frac{1}{\lambda_{\min}(\cG'+\eta \Lambda_o\otimes I_M)}\leq\left(\eta\lambda_2(L)+\min_{1\leq k \leq N}\lambda_{k,\min}\right)^{-1}.
\end{equation}
Since the $2-$induced norm of a positive definite matrix is equal to its maximum eigenvalue, we obtain:
\begin{equation}
\label{eq: bound relation a 2}
\|(\cG'+\eta \Lambda_o\otimes I_M)^{-1}\|\leq\left(\eta\lambda_2(L)+\min_{1\leq k \leq N}\lambda_{k,\min}\right)^{-1}.
\end{equation}
From the sub-multiplicative property of the $2-$induced norm and from~\eqref{eq: alternative K 22},~\eqref{eq: bound relation a 2}, and~\eqref{eq: bound relation a 3}, we obtain:
\begin{equation}
\|\cK\|\leq\|(\cG'+\eta \Lambda_o\otimes I_M)^{-1}\|\cdot\|\cG'\|\leq\left(\max_{1\leq k \leq N}\lambda_{k,\max}\right)\left(\eta\lambda_2(L)+\min_{1\leq k \leq N}\lambda_{k,\min}\right)^{-1}.
\end{equation}

%===================================
% Proof of convergence to a fixed point
%===================================
\section{Proof of Lemma~\ref{lem: fixed point convergence}}
\label{app: proof of convergence to a fixed point}
\noindent Given any two input vectors $\cx^1$ and $\cx^2$ with corresponding updated vectors $\cy^1$ and $\cy^2$, we have from~\eqref{eq: deterministic mapping}:
\begin{equation}
\label{eq: difference deterministic mapping}
\cy^1-\cy^2=\left(I_{MN}-\mu\eta\cL\right)\left(\cx^1-\cx^2-\mu\,\col\left\{\nabla_{w_k}J_k(x^1_k)-\nabla_{w_k}J_k(x^2_k)\right\}_{k=1}^N\right).
\end{equation}
From the mean-value theorem~\cite[pp.~24]{polyak1987introduction}, we have:
\begin{equation}
\label{eq: mean value theorem}
\nabla_{w_k}J_k(x^1_k)-\nabla_{w_k}J_k(x^2_k)=\left(\int_0^1\nabla^2_{w_k}J_k(x^2_k+t(x^1_k-x^2_k))dt\right)(x^1_k-x^2_k).
\end{equation}
Using~\eqref{eq: mean value theorem} into~\eqref{eq: difference deterministic mapping}, and the sub-multiplicative property of the $2-$induced norm~\cite{sayed2014adaptation}, we obtain:
\begin{equation}
\|\cy^1-\cy^2\|\leq\|I_{MN}-\mu\eta\cL\|\|\cD\|\|\cx^1-\cx^2\|,
\end{equation}
where
\begin{equation}
\label{eq: block diagonal D}
\mathcal{D}\triangleq\diag\left\{I_{M}-\mu\int_0^1\nabla^2_{w_k}J_k(x^2_k+t(x^1_k-x^2_k))dt\right\}_{k=1}^N.
\end{equation}
 We have
\begin{equation}
\|I_{MN}-\mu\eta\cL\|=\|(I_N-\mu\eta L)\otimes I_M\|=\|I_N-\mu\eta L\|.
\end{equation}
Let $\rho(\cdot)$ denote the spectral radius of its matrix argument. Since $L$ is symmetric, we have $\|I_N-\mu\eta L\|=\rho(I_N-\mu\eta L)$. Since $L$ has one eigenvalue at zero,  $\rho(I_N-\mu\eta L)$ is guaranteed to be equal to 1 if $\mu\eta$ satisfies condition~\eqref{eq: condition for stability}. For the  block diagonal symmetric matrix $\mathcal{D}$ in~\eqref{eq: block diagonal D}, we have:
\begin{equation}
\|\cD\|=\max_{1\leq k\leq N}\left\|I_{M}-\mu\int_0^1\nabla^2_{w_k}J_k(x^2_k+t(x^1_k-x^2_k))dt\right\|.
\end{equation}
Due to Assumption~\ref{assumption: strong convexity}, we have:
\begin{equation}
0<\lambda_{k,\min}I_M\leq\int_0^1\nabla^2_{w_k}J_k(x^2_k+t(x^1_k-x^2_k))dt\leq\lambda_{k,\max}I_M.
\end{equation}
It follows that $\|\cD\|\leq\gamma$ where $\gamma\triangleq\max_{1\leq k\leq N}\{\gamma_k\}$ and $\gamma_k$ is given in~\eqref{eq: gamma_k}. It holds that $0<\gamma_k<1$ when $\mu$ is chosen according to~\eqref{eq: condition 1}. Combining the previous results, we arrive at:
\begin{equation}
\|\cy^1-\cy^2\|_{2}\leq\gamma\|\cx^1-\cx^2\|_{2},
\end{equation}
for $\gamma<1$ when~\eqref{eq: condition for stability} and~\eqref{eq: condition 1} are satisfied and, in this case, the deterministic mapping~\eqref{eq: deterministic mapping} is a contraction.

%===============================
% App: Proof of the dimension of the bias
%==============================
\section{Proof of Theorem~\ref{theo: dimension of the bias}}
\label{app: proof of dimension of the bias}
\noindent From~\eqref{eq: recursion 2}, we obtain the following expression for $\cwt_{\infty}$:
\begin{equation}
\cwt_{\infty}=\mu\eta^2[\mathcal{H}_{\infty}+\eta\cL-\mu\eta\cL\mathcal{H}_{\infty}]^{-1} \cL^2 \cw^o_{\eta}\label{eq: ss network error vector}.
\end{equation}
Pre-multiplying both sides of~\eqref{eq: ss network error vector} by $\cV^{\top}= V^\top\otimes I_M$ gives:
\begin{equation}
\label{eq: GFT of the bias}
\cwb_{\infty}=\mu\eta^2\left[\cHb_{\infty}+\eta\cJ-\mu\eta\cJ\cHb_{\infty}\right]^{-1}\cJ^2\,\cwb^o_{\eta},
\end{equation}
where  $\cwb_{\infty}\triangleq\cV^\top\cwt_{\infty}$, {$\cwb^o_{\eta}\triangleq\cV^\top\cw^o_{\eta}$},
\begin{equation}
\cHb_{\infty}\triangleq\cV^\top\cH_{\infty}\cV,
\end{equation}
and $\cJ$ is given by~\eqref{eq: definition cJ}.

In the following we show that $\cwb_{\infty}$ can be written as:
\begin{equation}
\label{eq: wb infty}
\cwb_{\infty}=\mu\eta^2\left[
\begin{array}{c}
-\cP_{11}^{-1}\cP_{12}\cT\\
\cT
\end{array}
\right](\Lambda_o^2\otimes I_M)\cK[\cwb^o]_{2:N},
\end{equation}
where $\cK$ is defined in~\eqref{eq: K 22 lemma 1} and:
\begin{align}
\cT&\triangleq(\cP_{22}-\cP_{21}\cP_{11}^{-1}\cP_{12})^{-1},\label{eq: cT}\\
\label{eq: P_11}\cP_{11}&\triangleq(v_1^\top\otimes I_M)\cH_{\infty}(v_1\otimes I_M){=\frac{1}{N}\sum_{k=1}^NH_{k,\infty}},\\
\label{eq: P_12}\cP_{12}&\triangleq(v_1^\top\otimes I_M)\cH_{\infty}(V_R\otimes I_M),\\
\label{eq: P_21}\cP_{21}&\triangleq((I_{N-1}-\mu\eta\Lambda_o)\otimes I_M)(V_R^\top\otimes I_M)\cH_{\infty}(v_1\otimes I_M),\\
\label{eq: P_22}\cP_{22}&\triangleq\eta\Lambda_o\otimes I_M+((I_{N-1}-\mu\eta\Lambda_o)\otimes I_M)(V_R^\top\otimes I_M)\cH_{\infty}(V_R\otimes I_M),
\end{align}
{We introduce the following matrix, which appears in~\eqref{eq: GFT of the bias}:}
\begin{equation}
\cP\triangleq(I_{MN}-\mu\eta\cJ)\cHb_{\infty}+\eta\cJ=\left[\begin{array}{cc}
\cP_{11}&\cP_{12}\\
\cP_{21}&\cP_{22}
\end{array}
\right],
\end{equation}
where the blocks $\{\cP_{ij}\}$ are given by~\eqref{eq: P_11}--\eqref{eq: P_22}. {Note that, under Assumption~\ref{assumption: strong convexity}, $\cP_{11}$ in~\eqref{eq: P_11} is invertible since it can be bounded as follows:
\begin{equation}
\label{eq: bounded P11}
0<\frac{1}{N}\left(\sum_{k=1}^N\lambda_{k,\min}\right)I_M\leq\cP_{11}\leq\frac{1}{N}\left(\sum_{k=1}^N\lambda_{k,\max}\right)I_M.
\end{equation}
Applying} the block inversion formula to $\cP$, we obtain:
\begin{equation}
\label{eq: inverse of cP}
\cP^{-1}=\left[\begin{array}{cc}
\cP_{11}^{-1}+\cP_{11}^{-1}\cP_{12}\cT\cP_{21}\cP_{11}^{-1}&-\cP_{11}^{-1}\cP_{12}\cT\\
-\cT\cP_{21}\cP_{11}^{-1}&\cT
\end{array}
\right],
\end{equation}
with $\cT$ defined in~\eqref{eq: cT}. Replacing~\eqref{eq: inverse of cP} into~\eqref{eq: GFT of the bias}, and using~\eqref{eq: definition cJ} and~\eqref{eq: GFT of w o eta 2}, we conclude~\eqref{eq: wb infty}.

Our goal now is to show that
\begin{equation}
\lim_{\mu\rightarrow 0}\frac{\|\cw^o_{\eta}-\cw_{\infty}\|}{\mu}=c,
\end{equation}
for some constant $c$ that may depend on $\eta$ (the regularization strength), but not on $\mu$ (the step-size parameter). From~\eqref{eq: wb infty}, we have:
\begin{equation}
\label{eq: goal for the bias}
\lim_{\mu\rightarrow 0}\frac{\|\cw^o_{\eta}-\cw_{\infty}\|}{\mu}=\eta^2\lim_{\mu\rightarrow 0}\left\|\left[
\begin{array}{c}
-\cP_{11}^{-1}\cP_{12}\cT\\
\cT
\end{array}
\right](\Lambda_o^2\otimes I_M)\cK[\cwb^o]_{2:N}
\right\|.
\end{equation}
{Since the Euclidean norm is continuous, we have $\lim_{\mu\rightarrow 0}\| g(\mu)\|=\|\lim_{\mu\rightarrow 0}g(\mu)\|$.} In the following we show that
\begin{equation}
\label{eq: bound 1 on the bias}
\eta^4\left\|\lim_{\mu\rightarrow 0}\cT(\Lambda_o^2\otimes I_M)\cK[\cwb^o]_{2:N}\right\|^2\leq {O(\eta^4)(O(1)+O(\eta))^{-4}},
\end{equation}
and
\begin{equation}
\label{eq: bound 2 on the bias}
\eta^4\left\|\lim_{\mu\rightarrow 0}\cP_{11}^{-1}\cP_{12}\cT(\Lambda_o^2\otimes I_M)\cK[\cwb^o]_{2:N}\right\|^2\leq O(\eta^4)(O(1)+O(\eta))^{-4}.
\end{equation}
From~\eqref{eq: goal for the bias},~\eqref{eq: bound 1 on the bias}, and~\eqref{eq: bound 2 on the bias}, we can conclude~\eqref{eq: bound on the bias}.

%
%\begin{equation}
%\label{eq: bound 1 on the bias}
%\begin{split}
%\eta^4\left\|\lim_{\mu\rightarrow 0}\cT(\Lambda_o^2\otimes I_M)\cK[\cwb^o]_{2:N}\right\|^2&\leq\eta^4\lambda_{\max}^4(L)\left(\eta\lambda_2(L)+\min_{1\leq k \leq N}\lambda_{k,\min}\right)^{-4}\|[\cwb^o]_{2:N}\|^2\\
%&=O(\eta^4)(O(1)+O(\eta))^{-4},
%\end{split}
%\end{equation}

Let us first establish~\eqref{eq: bound 1 on the bias}. We have:
\begin{equation}
\label{eq: w12 Gmma 2}
\|\cT(\Lambda_o^2\otimes I_M)\cK[\cwb^o]_{2:N}\|^2\leq\|\cT\|^2\|\Lambda_o^2\|^2\|\cK\|^2\|[\cwb^o]_{2:N}\|^2.
\end{equation}
From~\eqref{eq: order of K 22}, we have $\|\cK\|^2\leq(O(1)+O(\eta))^{-2}$. For sufficiently small step-sizes, we have:
\begin{align}
&\lim_{\mu\rightarrow 0}\cP_{21}=(V_R^\top\otimes I_M)\cH_{\infty}(v_1\otimes I_M),\\
&\lim_{\mu\rightarrow 0}\cP_{22}=\eta\Lambda_o\otimes I_M+(V_R^\top\otimes I_M)\cH_{\infty}(V_R\otimes I_M).
\end{align}
Following the same line of reasoning as in~\eqref{eq: cG'}--\eqref{eq: bound relation a 2}, we can show that, when $\mu\rightarrow0$, we have:
\begin{equation}
\label{eq: order W22}
\|\cT\|^2\leq\left(\eta\lambda_2(L)+\min_{1\leq k \leq N}\lambda_{k,\min}\right)^{-2}=(O(1)+O(\eta))^{-2}.
\end{equation}
Thus, we conclude~\eqref{eq: bound 1 on the bias}.

Now, we establish~\eqref{eq: bound 2 on the bias}.   From~\eqref{eq: bounded P11}, we have $\cP_{11}=O(1)$ and $\|\cP_{11}^{-1}\|^2=O(1)$. Similarly, we can conclude from~\eqref{eq: P_12} that $\|\cP_{12}\|^2\leq O(1)$. Thus, using~\eqref{eq: bound 1 on the bias}, we arrive at~\eqref{eq: bound 2 on the bias}.

%===============================
% App: Proof of the dimension of the mean-square-error perturbation
%==============================
\section{Proof of Theorem~\ref{theo: dimension of the MSP}}
\label{app: proof of dimension of the MSP}
\noindent From~\eqref{eq: gradient noise process},~\eqref{eq: network vector recursion}, and~\eqref{eq: recursion 1}, we have:
\begin{equation}
\label{eq: ss bias of the algorithm}
\cw_{\infty}-\bcw_i=(I_{MN}-\mu\eta\cL)\left(\cw_{\infty}-\bcw_{i-1}-\mu\,\col\left\{\nabla_{w_k}J_k(w_{k,\infty})-\nabla_{w_k}J_k(\bw_{k,i-1})\right\}_{k=1}^N-\mu\,\col\left\{\bs_{k,i}(\bw_{k,i-1})\right\}_{k=1}^N\right).
\end{equation}
Using the mean-value theorem~\eqref{eq: mean value theorem}, the above relation can be written as:
\begin{equation}
\label{eq: ss bias of the algorithm 1}
\cw_{\infty}-\bcw_i=(I_{MN}-\mu\eta\cL)\left((I_{MN}-\mu\bcH_{i-1})(\cw_{\infty}-\bcw_{i-1})-\mu\,\col\{\bs_{k,i}(\bw_{k,i-1})\}_{k=1}^N\right),
\end{equation}
where $\bcH_{i-1}\triangleq\diag\{\bH_{1,i-1},\ldots,\bH_{N,i-1}\}$ with:
\begin{equation}
\label{eq: H_k i-1}
\bH_{k,i-1}\triangleq\int_{0}^1\nabla^2_{w_k}J_k(w_{k,\infty}-t(w_{k,\infty}-\bw_{k,i-1}))dt
\end{equation}
Let
\begin{align}
\label{eq: equation for bphi}\bphi_i&\triangleq(I_{MN}-\mu\bcH_{i-1})(\cw_{\infty}-\bcw_{i-1})-\mu\,\col\{\bs_{k,i}(\bw_{k,i-1})\}_{k=1}^N\\
C&\triangleq I_{N}-\mu\eta L.\label{eq: combination matrix C}
\end{align}
From the Laplacian matrix definition, it can be verified that the off-diagonal entries of the matrix $C$ are non-negative and that its diagonal entries are non-negative under condition~\eqref{eq: condition for positivity}. Furthermore, since we have $L\mathds{1}_N=0$, the entries on each row of $C$ will add up to one. Thus, applying Jensen's inequality~\cite[pp.~77]{boyd2004convex} {to} the convex function $\|\cdot\|^2$, we obtain from~\eqref{eq: ss bias of the algorithm 1} and~\eqref{eq: equation for bphi}:
\begin{equation}
\label{eq: relation by Jensen's}
\expec\|w_{k,\infty}-\bw_{k,i}\|^2\leq\sum_{\ell=1}^N[C]_{k\ell}\expec\|\bphi_{\ell,i}\|^2,
\end{equation}
where $\bphi_{k,i}$ is the $k$-th sub-vector of $\bphi_i$ given by:
\begin{equation}
\label{eq: equation for bphi_k}
\bphi_{k,i}=(I_M-\mu\bH_{k,i-1})(w_{k,\infty}-\bw_{k,i-1})-\mu\bs_{k,i}(\bw_{k,i-1}).
\end{equation}
{Squaring both sides of~\eqref{eq: equation for bphi_k}, conditioning on $\bcF_{i-1}$, and taking expectations we obtain:
\begin{equation}
\label{eq: relation conditioned}
\begin{split}
\expec[\|\bphi_{k,i}\|^2|\bcF_{i-1}]&=\|w_{k,\infty}-\bw_{k,i-1}\|^2_{\bSig_{k,i-1}}+\mu^2\expec[\|\bs_{k,i}(\bw_{k,i-1})\|^2|\bcF_{i-1}].
\end{split}
\end{equation}
where $\bSig_{k,i-1}\triangleq(I_M-\mu\bH_{k,i-1})^2$ and where the cross term is zero because of the zero-mean condition~\eqref{eq: mean gradient noise condition}}. Due to Assumption~\ref{assumption: strong convexity}, $\bSig_{k,i-1}$ can be bounded as follows:
\begin{equation}
\label{eq: bound on sigma_k, i-1}
0<\bSig_{k,i-1}\leq\gamma_k^2I_M,
\end{equation}
where $\gamma_k$ is given by~\eqref{eq: gamma_k}. From Assumption~\ref{assumption: gradient noise}, $\expec[\|\bs_{k,i}(\bw_{k,i-1})\|^2|\bcF_{i-1}]$ can be bounded as follows:
\begin{align}
\expec[\|\bs_{k,i}(\bw_{k,i-1})\|^2|\bcF_{i-1}]&\leq\beta_k^2\|\bw_{k,i-1}\|^2+\sigma^2_{s,k}\nonumber\\
&=\beta_k^2\|w^o_{k,\eta}-w_{k,\infty}+w_{k,\infty}-\bw_{k,i-1}-w^o_{k,\eta}\|^2+\sigma^2_{s,k}\nonumber\\
&\leq3\beta_k^2\|w^o_{k,\eta}-w_{k,\infty}\|^2+3\beta_k^2\|w_{k,\infty}-\bw_{k,i-1}\|^2+3\beta_k^2\|w^o_{k,\eta}\|^2+\sigma^2_{s,k}.\label{eq: gradient noise bounds second-moment}
\end{align}
Taking expectation again in~\eqref{eq: relation conditioned}, and using the bounds~\eqref{eq: bound on sigma_k, i-1} and~\eqref{eq: gradient noise bounds second-moment}, we obtain:
\begin{align}
\expec\|\bphi_{k,i}\|^2&=\expec\|w_{k,\infty}-\bw_{k,i-1}\|^2_{\bSig_{k,i-1}}+\mu^2\expec\|\bs_{k,i}(\bw_{k,i-1})\|^2\nonumber\\
&\leq(\gamma_k^2+3\mu^2\beta^2_k)\expec\|w_{k,\infty}-\bw_{k,i-1}\|^2+\mu^2\left(3\beta_k^2\|w^o_{k,\eta}-w_{k,\infty}\|^2+3\beta_k^2\|w^o_{k,\eta}\|^2+\sigma^2_{s,k}\right).\label{eq: relation unconditioned}
\end{align}
Now, combining~\eqref{eq: relation unconditioned} and~\eqref{eq: relation by Jensen's}, we obtain~\eqref{eq: evolution of the MSP i}.

Iterating~\eqref{eq: evolution of the MSP i} starting from $i=1$, we get:
\begin{equation}
\label{eq: recursion MSP 1}
\text{MSP}_i \preceq (CG)^i\text{MSP}_{0}+\mu^2\sum_{j=0}^{i-1}(CG)^jCb.
\end{equation}
Under Assumption~\ref{assumption: combination matrix} and condition~\eqref{eq: MSP stability}, the matrix $CG$ can be guaranteed to be stable. To see this, we upper bound the spectral radius as follows:
\begin{equation}
\rho(CG)\leq\|CG\|_{\infty}\leq\|C\|_{\infty}\|G\|_{\infty}=\|G\|_{\infty}=\max_{1\leq k\leq N}\gamma_k^2+3\mu^2\beta_k^2,
\end{equation}
where we used the fact that, under condition~\eqref{eq: condition for positivity}, the matrix $C$ is a right-stochastic matrix. {We have:
\begin{equation}
\label{eq: intermediate relation}
\gamma_k^2+3\mu^2\beta_k^2=\max\{1-2\mu\lambda_{k,\min}+\mu^2\lambda_{k,\min}^2+3\mu^2\beta_k^2,1-2\mu\lambda_{k,\max}+\mu^2\lambda_{k,\max}^2+3\mu^2\beta_k^2\},
\end{equation}
which is guaranteed to be less than one when:
\begin{equation}
0<\mu<\min\left\{\frac{2\lambda_{k,\min}}{\lambda_{k,\min}^2+3\beta_k^2},\frac{2\lambda_{k,\max}}{\lambda_{k,\max}^2+3\beta_k^2}\right\}.
\end{equation}
Then we conclude} that the matrix $CG$ is stable under condition~\eqref{eq: MSP stability}. In this case, we have:
\begin{equation}
\limsup_{i\rightarrow\infty}\text{MSP}_i\preceq\mu^2\sum_{j=0}^{\infty}(CG)^jCb.
\end{equation}
Using the submultiplicative property of the induced infinity norm, {we obtain:
\begin{equation}
\label{eq: limsup}
\begin{split}
\|\limsup_{i\rightarrow\infty}\text{MSP}_i\|_{\infty}&\leq\mu^2\left\|\sum_{j=0}^{\infty}(CG)^j\right\|_{\infty}\|C\|_{\infty}\|b\|_{\infty}\\
&\leq\mu^2\sum_{j=0}^{\infty}\|(CG)^j\|_{\infty}\|b\|_{\infty},\\
&\leq\mu^2\sum_{j=0}^{\infty}\|C\|^j_{\infty}\|G\|_{\infty}^j\|b\|_{\infty}=\frac{\mu^2\|b\|_{\infty}}{1-\|G\|_{\infty}},
\end{split}
\end{equation}
where} we used the fact that $\|C\|_{\infty}=1$ and where $\|G\|_{\infty}=\max_{1\leq k\leq N}\gamma_k^2+3\mu^2\beta_k^2$. {From~\eqref{eq: intermediate relation}, we have:
\begin{equation}
\gamma_k^2+3\mu^2\beta_k^2=1-\mu\zeta_k,
\end{equation}
where
\begin{equation}
\zeta_k\triangleq\min\{2\lambda_{k,\min}-\mu\lambda_{k,\min}^2-3\mu\beta_k^2, 2\lambda_{k,\max}-\mu\lambda_{k,\max}^2-3\mu\beta_k^2\}.
\end{equation}
Thus,
\begin{equation}
\|G\|_{\infty}=\max_{1\leq k\leq N}\{1-\mu\zeta_k\}=1-\mu\min_{1\leq k\leq N}\zeta_k.
%&=1-\mu\min_{1\leq k\leq N}\{2\lambda_{k,\min}-\mu\lambda_{k,\min}^2-3\mu\beta_k^2, 2\lambda_{k,\max}-\mu\lambda_{k,\max}^2-3\mu\beta_k^2\}.
\end{equation}
Substituting into~\eqref{eq: limsup}, we obtain:
\begin{equation}
\|\limsup_{i\rightarrow\infty}\text{MSP}_i\|_{\infty}\leq\frac{\mu\|b\|_{\infty}}{\min_{1\leq k\leq N}\zeta_k}.
\end{equation}
For} sufficiently small $\mu$, we have from~\eqref{eq: equation of b} and Theorem~\ref{theo: dimension of the bias} that $\|b\|=O(1)+O(\mu^2\eta^4)(O(1)+O(\eta))^{-4}$. We conclude that $\|\limsup_{i\rightarrow\infty}\text{MSP}_i\|_{\infty}{\leq} O(\mu)$.

From~\eqref{eq: bound on mean-square expectation}, we have:
\begin{align}
\limsup_{i\rightarrow\infty}\expec\|\cw^o_{\eta}-\bcw_i\|^2&\leq2\|\cw^o_{\eta}-\cw_{\infty}\|^2+2\limsup_{i\rightarrow\infty}\mathds{1}_N^\top\cdot\text{MSP}_i
\end{align}
Therefore, from Theorem~\ref{theo: dimension of the bias} and~\eqref{eq: steady-state of the MSP}, we conclude~\eqref{eq: steady-state of the second order}.
%===================================
% App: Proof of fourth-order error moment
%===================================
\section{Proof of Theorem~\ref{theo: dimension of the MFP}}
\label{app: proof of dimension of the MFP}
\noindent Applying Jensen's inequality~\cite[pp.~77]{boyd2004convex} {to} the convex function $\|\cdot\|^4$, we obtain from~\eqref{eq: ss bias of the algorithm 1} and~\eqref{eq: equation for bphi}:
\begin{equation}
\label{eq: relation by Jensen's 4}
\expec\|w_{k,\infty}-\bw_{k,i}\|^4\leq\sum_{\ell=1}^N[C]_{k\ell}\expec\|\bphi_{\ell,i}\|^4,
\end{equation}
where $C$ and $\bphi_{k,i}$ are given by~\eqref{eq: combination matrix C} and~\eqref{eq: equation for bphi_k}, respectively. Using the inequality~\cite[pp.~523]{sayed2014adaptation}:
\begin{equation}
\|a+b\|^4\leq\|a\|^4+3\|b\|^4+8\|a\|^2\|b\|^2+4\|a\|^2(a^{\top}b),
\end{equation}
we obtain from~\eqref{eq: equation for bphi_k} under Assumption~\ref{assumption: gradient noise} on the gradient noise:
\begin{equation}
\begin{split}
\expec\|\bphi_{k,i}\|^4\leq&~\expec\|(I_M-\mu\bH_{k,i-1})(w_{k,\infty}-\bw_{k,i-1})\|^4+3\mu^4\expec\|\bs_{k,i}(\bw_{k,i-1})\|^4+\\
&\quad8\mu^2\left(\expec\|(I_M-\mu\bH_{k,i-1})(w_{k,\infty}-\bw_{k,i-1})\|^2\right)\left(\expec\|\bs_{k,i}(\bw_{k,i-1})\|^2\right).
\end{split}
\end{equation}
From Assumption~\ref{assumption: strong convexity}, the matrices $(I_M-\mu\bH_{k,i-1})^2$ and $(I_M-\mu\bH_{k,i-1})^4$ can be bounded as follows:
\begin{equation}
0<(I_M-\mu\bH_{k,i-1})^2\leq\gamma_k^2I_M,
\end{equation}
\begin{equation}
0<(I_M-\mu\bH_{k,i-1})^4\leq\gamma_k^4I_M,
\end{equation}
where $\gamma_k$ is given by~\eqref{eq: gamma_k}. Thus, we obtain:
\begin{equation}
\label{eq: expec phi 4}
\expec\|\bphi_{k,i}\|^4\leq\gamma_k^4\expec\|w_{k,\infty}-\bw_{k,i-1}\|^4+3\mu^4\expec\|\bs_{k,i}(\bw_{k,i-1})\|^4+8\mu^2\gamma_k^2\left(\expec\|w_{k,\infty}-\bw_{k,i-1}\|^2\right)\left(\expec\|\bs_{k,i}(\bw_{k,i-1})\|^2\right).
\end{equation}
Under condition~\eqref{eq: condition on fourth-order moment of gradient noise}, {we have}:
\begin{align}
\expec\left[\|\bs_{k,i}(\bw_{k,i-1})\|^4|\bcF_{i-1}\right]&\leq\overline{\beta}_k^4\|\bw_{k,i-1}\|^4+\overline{\sigma}^4_{s,k}\nonumber\\
&=\overline{\beta}_k^4\|w^o_{k,\eta}-w_{k,\infty}+w_{k,\infty}-\bw_{k,i-1}-w^o_{k,\eta}\|^4+\overline{\sigma}^4_{s,k}\nonumber\\
&\leq27\overline{\beta}_k^4\|w^o_{k,\eta}-w_{k,\infty}\|^4+27\overline{\beta}_k^4\|w_{k,\infty}-\bw_{k,i-1}\|^4+27\overline{\beta}_k^4\|w^o_{k,\eta}\|^4+\overline{\sigma}^4_{s,k},\label{eq: MFP bound 1}
\end{align}
where we applied Jensen's inequality {to} the function $\|\cdot\|^4$. Furthermore, from~\eqref{eq: gradient noise bounds second-moment}, the last term on the {RHS} of~\eqref{eq: expec phi 4} can be bounded as follows:
\begin{align}
&\left(\expec\|w_{k,\infty}-\bw_{k,i-1}\|^2\right)\left(\expec\|\bs_{k,i}(\bw_{k,i-1})\|^2\right)\nonumber\\
&\leq3\beta_k^2\left(\expec\|w_{k,\infty}-\bw_{k,i-1}\|^2\right)^2+\left(3\beta_k^2\|w^o_{k,\eta}-w_{k,\infty}\|^2+3\beta_k^2\|w^o_{k,\eta}\|^2+\sigma^2_{s,k}\right)\expec\|w_{k,\infty}-\bw_{k,i-1}\|^2\nonumber\\
&\leq3\beta_k^2\expec\|w_{k,\infty}-\bw_{k,i-1}\|^4+\left(3\beta_k^2\|w^o_{k,\eta}-w_{k,\infty}\|^2+3\beta_k^2\|w^o_{k,\eta}\|^2+\sigma^2_{s,k}\right)\expec\|w_{k,\infty}-\bw_{k,i-1}\|^2,\label{eq: MFP bound 2}
\end{align}
where we used the fact that for any random variable $\ba$, we have $(\expec\ba)^2\leq\expec\ba^2$. Replacing~\eqref{eq: MFP bound 1} and~\eqref{eq: MFP bound 2} into~\eqref{eq: expec phi 4}, we obtain:
\begin{equation}
\label{eq: relation unconditioned 4}
\begin{split}
\expec\|\bphi_{k,i}\|^4\leq&~(\gamma_k^4+81\mu^4\overline{\beta}_k^4+24\mu^2\gamma_k^2\beta_k^2)\expec\|w_{k,\infty}-\bw_{k,i-1}\|^4+\\
&~8\mu^2\gamma_k^2(3\beta_k^2\|w^o_{k,\eta}-w_{k,\infty}\|^2+3\beta_k^2\|w^o_{k,\eta}\|^2+\sigma^2_{s,k})\expec\|w_{k,\infty}-\bw_{k,i-1}\|^2+\\
&~81\mu^4\overline{\beta}_k^4\|w^o_{k,\eta}-w_{k,\infty}\|^4+81\mu^4\overline{\beta}_k^4\|w^o_{k,\eta}\|^4+3\mu^4\overline{\sigma}^4_{s,k}.
\end{split}
\end{equation}
Now, combining~\eqref{eq: relation unconditioned 4} and~\eqref{eq: relation by Jensen's 4}, we arrive at~\eqref{eq: evolution of the MFP i}.

Iterating~\eqref{eq: evolution of the MFP i} starting from $i=1$, we get:
\begin{equation}
\label{eq: recursion MFP 1}
\text{MFP}_i \preceq (CG')^i\text{MFP}_{0}+\mu^2\sum_{j=0}^{i-1}(CG')^jCB\text{MSP}_{i-1-j}+\mu^4\sum_{j=0}^{i-1}(CG')^jCb'.
\end{equation}
Under Assumption~\ref{assumption: combination matrix} and for sufficiently small $\mu$, the matrix $CG'$ can be guaranteed to be stable. To see this, we upper bound its spectral radius as follows:
\begin{equation}
\rho(CG')\leq\|CG'\|_{\infty}\leq\|C\|_{\infty}\|G'\|_{\infty}=\|G'\|_{\infty},
\end{equation}
since under condition~\eqref{eq: condition for positivity}, $C$ is a right-stochastic matrix. The $\infty-$norm of $G'$ is given {by:
\begin{align}
\|G'\|_{\infty}&=\max_{1\leq k\leq N}\left\{\gamma_k^4+24\mu^2\gamma_k^2\beta_k^2+81\mu^4\overline{\beta}_k^4\right\}\nonumber\\
&=\max_{1\leq k\leq N}\big\{\max\{1-4\mu\lambda_{k,\min}+6\mu^2\lambda_{k,\min}^2-4\mu^3\lambda_{k,\min}^3+\mu^4\lambda_{k,\min}^4+24\mu^2\gamma_k^2\beta_k^2+81\mu^4\overline{\beta}_k^4,\nonumber\\
&~\quad\qquad\qquad\qquad1-4\mu\lambda_{k,\max}+6\mu^2\lambda_{k,\max}^2-4\mu^3\lambda_{k,\max}^3+\mu^4\lambda_{k,\max}^4+24\mu^2\gamma_k^2\beta_k^2+81\mu^4\overline{\beta}_k^4\}\big\}\nonumber\\
&=1-\mu\min_{1\leq k\leq N}\big\{\min\{4\lambda_{k,\min}-6\mu\lambda_{k,\min}^2+4\mu^2\lambda_{k,\min}^3-\mu^3\lambda_{k,\min}^4-24\mu\gamma_k^2\beta_k^2-81\mu^3\overline{\beta}_k^4,\nonumber \\
&~\quad\qquad\qquad\qquad\qquad4\lambda_{k,\max}-6\mu\lambda_{k,\max}^2+4\mu^2\lambda_{k,\max}^3-\mu^3\lambda_{k,\max}^4-24\mu\gamma_k^2\beta_k^2-81\mu^3\overline{\beta}_k^4\}\big\}.\label{eq: G' infty}
\end{align}
A} sufficiently small $\mu$ ensures $\|G'\|_{\infty}<1$ and, thus, ensures the stability of $CG'$.

We have established in Theorem~\ref{theo: dimension of the MSP} that, for small $\mu$, after sufficient iterations have passed, $\text{MSP}_j$ converges to a bounded region on the order of $\mu$. This implies that, there exists a $j_o$ large enough such that for all $j\geq j_o$ it holds that:
\begin{equation}
\label{eq: MSP after j_o}
\|\text{MSP}_j\|_{\infty}\leq s_{\max}=O(\mu).
\end{equation}
In this case, we have from~\eqref{eq: recursion MFP 1}:
\begin{align}
\limsup_{i\rightarrow\infty}\text{MFP}_i&\preceq\mu^4\sum_{j=0}^{\infty}(CG')^jCb'+\mu^2\limsup_{i\rightarrow\infty}\sum_{j=0}^{i-1}(CG')^jCB\text{MSP}_{i-1-j}\nonumber\\
&=\mu^4\sum_{j=0}^{\infty}(CG')^jCb'+\mu^2\limsup_{i\rightarrow\infty}\sum_{j=0}^{i-1}(CG')^{i-1-j}CB\text{MSP}_{j}\nonumber\\
&=\mu^4\sum_{j=0}^{\infty}(CG')^jCb'+\mu^2\limsup_{i\rightarrow\infty}\left(\sum_{j=0}^{j_o}(CG')^{i-1-j}CB\text{MSP}_{j}+\sum_{j=j_o+1}^{i-1}(CG')^{i-1-j}CB\text{MSP}_{j}\right)\nonumber\\
&=\mu^4\sum_{j=0}^{\infty}(CG')^jCb'+\mu^2\limsup_{i\rightarrow\infty}\left((CG')^{i}\sum_{j=0}^{j_o}(CG')^{-1-j}CB\text{MSP}_{j}+\sum_{j=j_o+1}^{i-1}(CG')^{i-1-j}CB\text{MSP}_{j}\right)\nonumber\\
&=\mu^4\sum_{j=0}^{\infty}(CG')^jCb'+\mu^2\limsup_{i\rightarrow\infty}\sum_{j=j_o+1}^{i-1}(CG')^{i-1-j}CB\text{MSP}_{j}.
\end{align}
Using the submultiplicative and sub-additive properties of the induced infinity norm, we obtain:
\begin{align}
\|\limsup_{i\rightarrow\infty}\text{MFP}_i\|_{\infty}&\leq\mu^4\left\|\sum_{j=0}^{\infty}(CG')^j\right\|_{\infty}\|b'\|_{\infty}+\mu^2\limsup_{i\rightarrow\infty}\sum_{j=j_o+1}^{i-1}\|(CG')^{i-1-j}CB\text{MSP}_{j}\|_{\infty}\nonumber\\
&\leq\mu^4\sum_{j=0}^{\infty}\|G'\|_{\infty}^j\|b'\|_{\infty}+\mu^2\limsup_{i\rightarrow\infty}\sum_{j=0}^{i-j_o-2}\|G'\|_{\infty}^{j}\|B\|_{\infty}s_{\max}\nonumber\\
&=\mu^4\frac{\|b'\|_{\infty}}{1-\|G'\|_{\infty}}+\mu^2\frac{\|B\|_{\infty}s_{\max}}{1-\|G'\|_{\infty}}\label{eq: limsup 1}
\end{align}
where in the second line we used~\eqref{eq: MSP after j_o} and where $\|G'\|_{\infty}$ is given by~\eqref{eq: G' infty}. Since $\|b'\|_{\infty}=O(1)$, $\|B\|_{\infty}=O(1)$, $s_{\max}=O(\mu)$, and $1-\|G'\|_{\infty}=O(\mu)$, we conclude~\eqref{eq: steady-state of the MFP i}.

From~\eqref{eq: fourth order bound 1} and~\eqref{eq: fourth order bound 2}, we have:
\begin{equation}
\limsup_{i\rightarrow\infty}\expec\|\cw_{\eta}^o-\bcw_{i}\|^4\leq 8\|\cw_{\eta}^o-\cw_{\infty}\|^4+8N\limsup_{i\rightarrow\infty}\mathds{1}_N\cdot\text{MFP}_i.
\end{equation}
Therefore, from Theorem~\ref{theo: dimension of the bias} and~\eqref{eq: steady-state of the MFP i}, we conclude~\eqref{eq: steady-state of the fourth order}.

%=================================
% App: Proof of mean stability
%=================================
\section{Proof of Theorem~\ref{theo: dimension of the SMP}}
\label{app: proof of dimension of the SMP}
\noindent Conditioning both sides of~\eqref{eq: ss bias of the algorithm 1}, invoking the conditions on the gradient noise from Assumption~\ref{assumption: gradient noise}, and computing the conditional expectations, we obtain:
\begin{equation}
\expec[(\cw_{\infty}-\bcw_i)|\bcF_{i-1}]=(I_{MN}-\mu\eta\cL)(I_{MN}-\mu\bcH_{i-1})(\cw_{\infty}-\bcw_{i-1}).
\end{equation}
Taking expectations again, we arrive at:
\begin{equation}
\expec(\cw_{\infty}-\bcw_i)=(I_{MN}-\mu\eta\cL)\expec[(I_{MN}-\mu\bcH_{i-1})(\cw_{\infty}-\bcw_{i-1})].
\end{equation}
Applying Jensen's inequality~\cite[pp.~77]{boyd2004convex} {to} the convex function $\|\cdot\|^2$, we obtain from the above relation:
\begin{equation}
\label{eq: Jensen's inequality for first order}
\|\expec (w_{k,\infty}-\bw_{k,i})\|^2\leq\sum_{\ell=1}^N[C]_{k\ell}\left\|\expec\left[(I_M-\mu\bH_{\ell,i-1}) (w_{\ell,\infty}-\bw_{\ell,i-1})\right]\right\|^2,
\end{equation}
where $C$ and $\bH_{k,i-1}$ are given by~\eqref{eq: combination matrix C} and~\eqref{eq: H_k i-1}, respectively.
%\begin{equation}
%\bH_{k,i-1}\triangleq\int_{0}^1\nabla^2_{w_k}J_k(w_{k,\infty}-t(w_{k,\infty}-\bw_{k,i-1}))dt.
%\end{equation}
Let
\begin{equation}
\widetilde{\bH}_{k,i-1}\triangleq H_{k,\eta}-\bH_{k,i-1},
\end{equation}
where
\begin{equation}
H_{k,\eta}=\nabla_{w_k}^2J_k(w_{k,\eta}^o).
\end{equation}
Then, we can write:
\begin{equation}
\expec\left[(I_M-\mu\bH_{k,i-1}) (w_{k,\infty}-\bw_{k,i-1})\right]=(I_M-\mu H_{k,\eta}) \expec(w_{k,\infty}-\bw_{k,i-1})+\mu\bc_{k,i-1},
\end{equation}
in terms of a deterministic perturbation sequence defined by
\begin{equation}
\bc_{k,i-1}\triangleq\expec[\widetilde{\bH}_{k,i-1} (w_{k,\infty}-\bw_{k,i-1})].
\end{equation}
By applying Jensen's inequality {to} the convex function $\|\cdot\|^2$, we obtain:
\begin{align}
&\|\expec(I_M-\mu\bH_{k,i-1}) (w_{k,\infty}-\bw_{k,i-1})\|^2\nonumber\\
&=\|(I_M-\mu H_{k,\eta}) \expec(w_{k,\infty}-\bw_{k,i-1})+\mu\bc_{k,i-1}\|^2\nonumber\\
&=\left\|t\frac{1}{t}(I_M-\mu H_{k,\eta})\expec(w_{k,\infty}-\bw_{k,i-1})+\mu(1-t)\frac{1}{1-t} \bc_{k,i-1}\right\|^2\nonumber\\
&\leq t\left\|\frac{1}{t}(I_M-\mu H_{k,\eta})\expec(w_{k,\infty}-\bw_{k,i-1})\right\|^2+\mu^2(1-t)\left\|\frac{1}{1-t} \bc_{k,i-1}\right\|^2\nonumber\\
&= \frac{1}{t}\|(I_M-\mu H_{k,\eta})\expec(w_{k,\infty}-\bw_{k,i-1})\|^2+\mu^2\frac{1}{1-t}\| \bc_{k,i-1}\|^2,
\end{align}
for any arbitrary positive number $t\in(0,1)$. We select $t=\gamma_k$ where $\gamma_k$ is given by~\eqref{eq: gamma_k}, which is guaranteed to be less than one under condition~\eqref{eq: condition 1}. From Assumption~\ref{assumption: strong convexity}, we have $\|I_M-\mu H_{k,\eta}\|^2\leq\gamma^2_k$. Thus, we obtain:
\begin{equation}
\label{eq: bound on square norm phi k}
\|\expec(I_M-\mu\bH_{k,i-1}) (w_{k,\infty}-\bw_{k,i-1})\|^2\leq \gamma_k\|\expec(w_{k,\infty}-\bw_{k,i-1})\|^2+\mu^2\frac{1}{1-\gamma_k}\| \bc_{k,i-1}\|^2.
\end{equation}
As shown in~\cite[Appendix~E]{sayed2014adaptation}, the Hessian of a twice differentiable strongly convex function $J_{k}(w_k)$ satisfying Assumptions~\ref{assumption: strong convexity} and~\ref{assumption: local smoothness hessian} is globally Lipschitz relative to $w^o_{k,\eta}$, namely, it satisfies:
\begin{equation}
\label{eq: global Lipschitz condition}
\|\nabla^2_{w_k}J_k(w_k)-\nabla^2_{w_k}J_k(w^o_{k,\eta})\|\leq \kappa'_d\|w_k-w^o_{k,\eta}\|,\quad\forall w_k,
\end{equation}
where $\kappa'_d=\max\{\kappa_d,\frac{\lambda_{k,\max}-\lambda_{k,\min}}{\epsilon}\}$. Then, for each agent $k$ we obtain:
\begin{align}
\|\widetilde{\bH}_{k,i-1}\|&\triangleq\|H_{k,\eta}-\bH_{k,i-1}\|\nonumber\\
&\leq\int_{0}^1\left\|\nabla_{w_k}^2J_k(w_{k,\eta}^o)-\nabla_{w_k}^2J_k(w_{k,\infty}-t(w_{k,\infty}-\bw_{k,i-1}))\right\|dt\nonumber\\
&\leq\int_{0}^1\kappa'_d\|w_{k,\eta}^o-w_{k,\infty}+t(w_{k,\infty}-\bw_{k,i-1})\|dt\nonumber\\
&\leq\int_{0}^1\kappa'_d\|w_{k,\eta}^o-w_{k,\infty}\|dt+\int_{0}^1\kappa'_d\|t(w_{k,\infty}-\bw_{k,i-1})\|dt\nonumber\\
&=\kappa'_d\|w_{k,\eta}^o-w_{k,\infty}\|+\frac{1}{2}\kappa'_d\|w_{k,\infty}-\bw_{k,i-1}\|,
\end{align}
{and, hence,}
\begin{align}
\|\bc_{k,i-1}\|&=\|\expec[\widetilde{\bH}_{k,i-1} (w_{k,\infty}-\bw_{k,i-1})]\|\nonumber\\
&\leq\expec[\|\widetilde{\bH}_{k,i-1}\|\|w_{k,\infty}-\bw_{k,i-1}\|]\nonumber\\
&\leq\kappa'_d\|w_{k,\eta}^o-w_{k,\infty}\|\expec\|w_{k,\infty}-\bw_{k,i-1}\|+\frac{1}{2}\kappa'_d\expec\|w_{k,\infty}-\bw_{k,i-1}\|^2.\label{eq: norm of c k}
\end{align}
where we used the stochastic version of Jensen's inequality:
\begin{equation}
f(\expec\ba)\leq\expec(f(\ba))
\end{equation}
when $f(x)\in\mathbb{R}$ is convex.
%It then follows that:
%\begin{equation}
%\limsup_{i\rightarrow\infty}\|\bc_{k,i-1}\|=O(\mu).
%\end{equation}
Applying Jensen's inequality {to} the convex function $\|\cdot\|^2$ and using the fact that $(\expec\ba)^2\leq\expec\ba^2$ for any real-valued random variable $\ba$, we obtain from~\eqref{eq: norm of c k}:
\begin{align}
\|\bc_{k,i-1}\|^2&\leq 2(\kappa'_d)^2\|w_{k,\eta}^o-w_{k,\infty}\|^2\left(\expec\|w_{k,\infty}-\bw_{k,i-1}\|\right)^2+{2\frac{1}{4}}(\kappa'_d)^2\left(\expec\|w_{k,\infty}-\bw_{k,i-1}\|^2\right)^2\nonumber\\
&\leq2(\kappa'_d)^2\|w_{k,\eta}^o-w_{k,\infty}\|^2\expec\|w_{k,\infty}-\bw_{k,i-1}\|^2+\frac{1}{2}(\kappa'_d)^2\expec\|w_{k,\infty}-\bw_{k,i-1}\|^4.
\end{align}
From~\eqref{eq: Jensen's inequality for first order} and using the above bound in~\eqref{eq: bound on square norm phi k}, we conclude~\eqref{eq: evolution of the SMP i}.

Iterating~\eqref{eq: evolution of the SMP i} starting from $i=1$, we obtain:
\begin{equation}
\text{SMP}_i\preceq (CG'')^i\text{SMP}_{0}+\mu^2\frac{1}{2}(\kappa'_d)^2\sum_{j=0}^{i-1}(CG'')^jC(I-G'')^{-1}\text{MFP}_{i-1-j}+\mu^2\sum_{j=0}^{i-1}(CG'')^jC(I-G'')^{-1}B'\text{MSP}_{i-1-j}.
\end{equation}
Under Assumption~\ref{assumption: combination matrix} and condition~\eqref{eq: condition 1}, the matrix $CG''$ is guaranteed to be stable.  From~\eqref{eq: steady-state of the MSP},~\eqref{eq: steady-state of the MFP i},~and following similar arguments as the {ones} used to establish~\eqref{eq: steady-state of the MFP i} in Appendix~\ref{app: proof of dimension of the MFP}, we conclude that
\begin{equation}
\|\limsup_{i\rightarrow\infty}\text{SMP}_i\|_{\infty}=O(\mu^2)+\frac{O(\mu^3\eta^4)}{(O(1)+O(\eta))^4}=O(\mu^2),
\end{equation}
where we used the fact that $\|B'\|_{\infty}\leq O(\mu^2\eta^4)/(O(1)+O(\eta))^4$ from Theorem~\ref{theo: dimension of the bias} and $\|(I_N-G'')^{-1}\|_{\infty}\leq O(\mu^{-1})$.

Using~\eqref{eq: triangle inequality} and since $\|\expec(\cw_{\infty}-\bcw_{i})\|^2=\mathds{1}_N\cdot\text{SMP}_i$, we conclude~\eqref{eq: steady-state of the first order} from Theorem~\ref{theo: dimension of the bias} and~\eqref{eq: steady-state of the SMP i}.

\end{appendices}

%=====================================
% Sec: References
%=====================================
\bibliographystyle{IEEEbib}
\bibliography{reference}

\end{document}